%% file: main.tex
\documentclass[sigconf]{metafiles/acmart}
\AtBeginDocument{%
  \providecommand\BibTeX{{%
    \normalfont B\kern-0.5em{\scshape i\kern-0.25em b}\kern-0.8em\TeX}}}

\copyrightyear{2024} 
\acmYear{2024} 
\setcopyright{rightsretained} 
\acmConference[KDD '24]{Proceedings of the 30th ACM SIGKDD Conference on
Knowledge Discovery and Data Mining}{August 25--29, 2024}{Barcelona, Spain}
\acmBooktitle{Proceedings of the 30th ACM SIGKDD Conference on Knowledge
Discovery and Data Mining (KDD '24), August 25--29, 2024, Barcelona,
Spain}
\acmDOI{10.1145/3637528.3672046}
\acmISBN{979-8-4007-0490-1/24/08}

\usepackage[utf8]{inputenc} 
\usepackage[T1]{fontenc}    
\usepackage{array,url,booktabs,microtype,cite,multirow,hyperref,subcaption,graphicx,mathtools,epstopdf,xspace,amsthm,amsfonts,nicefrac,xcolor,makecell,adjustbox}
\usepackage{wrapfig,lipsum,booktabs}
\usepackage{metafiles/purudefs}
\usepackage{tabularx}
\usepackage{balance}
\usepackage{soul}
\usepackage{algorithm}
\usepackage{algpseudocode}
\usepackage{algorithmicx}
\usepackage{enumitem}

\usepackage{enumitem}
\input{metafiles/definitions}

\input{metafiles/math_commands}

\begin{document}

\title{\algtitle}

\author{Sachin Yadav}
\authornotemark[1]\authornotemark[2]\authornotemark[3]
\email{t-sacyadav@microsoft.com}
\affiliation{%
  \institution{Microsoft Research}
  \country{Bengaluru, India}
}

\author{Deepak Saini}
\authornote{Equal technical contribution. \,\(^{\dagger}\)Sachin Yadav led the project. \\ \(^{\ddagger}\)Corresponding authors are S. Yadav (\href{mailto:sachinyadav7024@gmail.com}{sachinyadav7024@gmail.com}) and S. Asokan (\href{mailto:sasokan@microsoft.com}{sasokan@microsoft.com})}
\email{desaini@microsoft.com}
\affiliation{%
  \institution{Microsoft}
  \country{Redmond, USA}
}

\author{Anirudh Buvanesh}
\authornotemark[1]
\email{t-abuvanesh@microsoft.com}
\affiliation{%
  \institution{Microsoft Research}
  \country{Bengaluru, India}
}

\author{Bhawna Paliwal}
\email{bhawna@microsoft.com}
\affiliation{%
  \institution{Microsoft Research}
  \country{Bengaluru, India}
}

\author{Kunal Dahiya}
\email{kunalsdahiya@gmail.com}
\affiliation{%
  \institution{Indian Institute of Technology}
  \country{Delhi, India}
}

\author{Siddarth Asokan}
\authornotemark[3]
\email{sasokan@microsoft.com}
\affiliation{%
  \institution{Microsoft Research}
  \country{Bengaluru, India}
}

\author{Yashoteja Prabhu}
\email{yprabhu@microsoft.com}
\affiliation{%
  \institution{Microsoft Research}
  \country{Bengaluru, India}
}

\author{Jian Jiao}
\email{Jian.Jiao@microsoft.com}
\affiliation{%
  \institution{Microsoft}
  \country{Redmond, USA}
}

\author{Manik Varma}
\email{manik@microsoft.com}
\affiliation{%
  \institution{Microsoft Research}
  \country{Bengaluru, India}
}

\renewcommand{\shortauthors}{Sachin Yadav et al.}


\keywords{extreme classification, large-scale retrieval, zero-shot retrieval, sponsored search advertising}

\input{sections/1_abstract}
\maketitle
\input{sections/2_introduction}

\input{sections/3_related_work}
\input{sections/4_method}
\input{sections/5_implementation}
\input{sections/6_theory}
\input{sections/7_experiments}
\input{sections/8_conclusions}

\clearpage
\bibliographystyle{metafiles/ACM-Reference-Format}
\balance
\bibliography{references}

\input{sections/9_appendix}
\input{sections/9_ethics}
\end{document}

%% file: metafiles/definitions.tex
\newtheorem*{assumption*}{\assumptionnumber}

\renewcommand{\it}[1]{\textit{#1}}

\newcommand{\new}{{novel}\xspace}
\newcommand{\unseen}{{novel}\xspace}
\newcommand{\seen}{{observed}\xspace}
\newcommand{\alg}{{IRENE}\xspace}

\newcommand{\xmec}{{EMMETT}\xspace}
\newcommand{\algtitle}{Extreme Meta-Classification for Large-Scale Zero-Shot Retrieval}

\newcommand{\ova}{1-vs-all\xspace}

\definecolor{arsenic}{rgb}{0.23, 0.27, 0.29}

\definecolor{silver}{rgb}{0.75, 0.75, 0.75}

\definecolor{ForestGreen}{HTML}{228B22}
\definecolor{BrightBlue}{rgb}{0., 0., 1.0} 
\definecolor{DarkBlue}{rgb}{0.0, 0.0, 0.5}   

\newcommand{\BBlue}[1]{{\textcolor{BrightBlue}{#1}}}

%% file: metafiles/math_commands.tex
\usepackage{amsmath,amsfonts,bm}
\usepackage{mathrsfs} 
\usepackage{amsthm} 
\usepackage{rotating}
\newtheorem*{theorem*}{Theorem}
\newtheorem*{lemma*}{Lemma}
\newtheorem*{corollary*}{Corollary}
\newcolumntype{P}[1]{>{\centering\arraybackslash}p{#1}} 

\newcommand{\bbE}{{\mathbb{E}}}

\def\eqref#1{equation~(\ref{#1})}
\def\Eqref#1{Equation~(\ref{#1})}

\def\1{\bm{1}}

\usepackage{amsmath}
\makeatletter
\newcommand{\subalign}[1]{%
  \vcenter{%
    \Let@ \restore@math@cr \default@tag
    \baselineskip\fontdimen10 \scriptfont\tw@
    \advance\baselineskip\fontdimen12 \scriptfont\tw@
    \lineskip\thr@@\fontdimen8 \scriptfont\thr@@
    \lineskiplimit\lineskip
    \ialign{\hfil$\m@th\scriptstyle##$&$\m@th\scriptstyle{}##$\hfil\crcr
      #1\crcr
    }%
  }%
}
\makeatother

\def\vb{{\bm{b}}}

\def\vs{{\bm{s}}}
\def\vt{{\bm{t}}}
\def\vu{{\bm{u}}}
\def\vv{{\bm{v}}}
\def\vw{{\bm{w}}}
\def\vx{{\bm{x}}}

\def\vz{{\bm{z}}}

\DeclareMathAlphabet{\mathsfit}{\encodingdefault}{\sfdefault}{m}{sl}
\SetMathAlphabet{\mathsfit}{bold}{\encodingdefault}{\sfdefault}{bx}{n}

\def\gE{{\mathcal{E}}}

\def\gG{{\mathcal{G}}}

\def\gL{{\mathcal{L}}}

\def\gS{{\mathcal{S}}}

\def\gX{{\mathcal{X}}}
\def\gY{{\mathcal{Y}}}
\def\gZ{{\mathcal{Z}}}

\def\sR{{\mathbb{R}}}
\def\sS{{\mathbb{S}}}

\def\sW{{\mathbb{W}}}
\def\sX{{\mathbb{X}}}

\def\sZ{{\mathbb{Z}}}



%% file: sections/1_abstract.tex
\begin{abstract}
We develop accurate and efficient solutions for large-scale retrieval tasks where novel (\textit{zero-shot}) items can arrive continuously at a rapid pace. Conventional Siamese-style approaches embed both queries and items through a small encoder and retrieve the items lying closest to the query. While this approach allows efficient addition and retrieval of novel items, the small encoder lacks sufficient capacity for the necessary world knowledge in complex retrieval tasks. The extreme classification approaches have addressed this by learning a separate classifier for each item observed in the training set which significantly increases the representation capacity of the model. Such classifiers outperform Siamese approaches on observed items, but cannot be trained for novel items due to data and latency constraints. To bridge these gaps, this paper develops: (1) A \textbf{new algorithmic framework}, EMMETT, which efficiently synthesizes classifiers on-the-fly for novel items, by relying on the readily available classifiers for observed items; (2) A \textbf{new algorithm}, IRENE, which is a simple and effective instance of EMMETT that is specifically suited for large-scale deployments, and (3) A \textbf{new theoretical framework} for analyzing the generalization performance in large-scale zero-shot retrieval which guides our algorithm and training related design decisions.

Comprehensive experiments are conducted on a wide range of retrieval tasks which demonstrate that IRENE improves the zero-shot retrieval accuracy by up to 15\% points in Recall@10 when added on top of leading encoders. Additionally, on an online A/B test in a large-scale ad retrieval task in a major search engine, IRENE improved the ad click-through rate by 4.2\%. Lastly, we validate our design choices through extensive ablative experiments. The source code for IRENE is available at \BBlue{\url{https://aka.ms/irene}}.

\end{abstract}

%% file: sections/2_introduction.tex
\section{Introduction}

Large-scale retrieval involves retrieving the items relevant to a query from a pool of hundreds of millions of candidate items. Such tasks frequently arise in modern-day web applications such as web search~\citep{Bajaj18}, computational advertising~\citep{Agrawal13}, product recommendation~\citep{Chang21}, and so on. To optimize user satisfaction, the retrieved results need to be highly relevant to the query and delivered in real-time, typically within milliseconds. Additionally, due to the exponential growth of digital content, novel items get introduced into these systems in vast quantities daily, which need to be swiftly processed and inserted into the candidate pool to ensure up-to-date results. This paper aims to develop highly accurate and efficient solutions for problems of large-scale text-based retrieval with \textit{novel} items (also referred to as \textit{zero-shot} items). Note that this scenario is different from the one considered in ~\citep{Gao23,Xin22} where the task itself is zero-shot.

A widely used technique for large-scale retrieval is dense retrieval~\citep{Zhao23} where both queries and items are represented as embeddings in a shared low-dimensional space such that the items relevant to a query are positioned closer to it than the irrelevant ones. The relevant items, typically few in number, are then retrieved in almost real-time using scalable Approximate Nearest Neighbour Search (ANNS) indices~\citep{jayaram2019diskann}. The accuracy of retrieved results depends on the quality of the derived query and item representations.

Typical dense retrievers are based on a Siamese encoder architecture which uses a common deep neural encoder to derive both query and item representations from their raw text inputs~\citep{Xiong20,Karpukhin20}. Usually, a small and efficient encoder is utilized to reduce the representation latencies, which enables large-scale deployments and quick insertion of novel items. However, a small encoder often lacks the capacity to model the complexity inherent in retrieval tasks~\citep{Dahiya23b}. For instance, item descriptions frequently contain named entities, numbers, model numbers, and ambiguous phrases with issues of synonymy and polysemy whose resolution requires extensive world knowledge that cannot be contained within a small encoder~\citep{query2doc23}. As a result, the representations from these approaches are inferior and degrade the retrieval performance.

Recently, Extreme Classification (XC) methods have emerged as promising alternatives for large scale retrieval. Leading extreme classifiers augment a small Siamese encoder with a massive linear classifier layer at its output which significantly boosts the model capacity~\citep{Dahiya21b,Dahiya23}. Each item \textit{observed} in the training set is endowed with its own classifier, which absorbs the world knowledge pertinent to the item when trained from the historical click logs. During retrieval, a query is first passed through the encoder and then projected onto its relevant items by applying classifiers. When there are enough clicked query samples for training, the classifier-based item representations can be more precise than the text-restricted representations from a small encoder. However, zero-shot retrieval is not supported, as classifiers cannot be trained for novel items with no clicked samples. Even if we could, learning new classifiers from scratch is time-consuming and delays the representation of novel items, making the candidate pool outdated.

This paper addresses the limitations of the existing approaches on large-scale zero-shot retrieval. Our primary research question is: \textit{\textbf{How to construct accurate representations for novel items without significant computational overhead in large-scale retrieval tasks?}}

We propose a novel algorithmic framework, \ul{E}xtre\ul{M}e \ul{MET}a-classifica\ul{T}ion (\xmec), as an answer to this question. \xmec extends conventional Extreme Classification to zero-shot scenarios. Often in large-scale retrieval tasks, millions of observed items with associated user-clicked queries are available for training from historical logs. This leads to two key insights that underpin \xmec's design: 

First, ignoring any significant shifts in item distribution, extreme classifiers trained for a million observed items are likely to contain most of the world knowledge required for any novel item in a distilled and readily usable form. With this intuition, \xmec builds a pipeline with two main modules: (1) A classifier selector (\(\cS\)) module, which takes a novel item as input and rapidly shortlists a few observed item classifiers that are most informative for it, by using efficient filtering techniques, and (2) A meta-classifier generator (\(\cG\)) module, which combines these shortlisted classifiers to synthesize a novel item's classifier with minimal latency. We call such a classifier derived from other classifiers as a \textit{meta-classifier}. 
Second, a million observed items can also serve as a large number of samples for optimally training the \xmec modules. With a careful loss design, each observed item can be treated as a proxy for a zero-shot item which yields a massive training set for zero-shot retrieval. Training \xmec on such a dataset can ensure robust generalization on novel items. We theoretically and empirically validate this claim in Sections \textbf{\ref{Sec:Theory}} and \textbf{\ref{Sec:Ablations}}, respectively.

\begin{figure}[!t]
\centering
\includegraphics[width=0.98\linewidth]{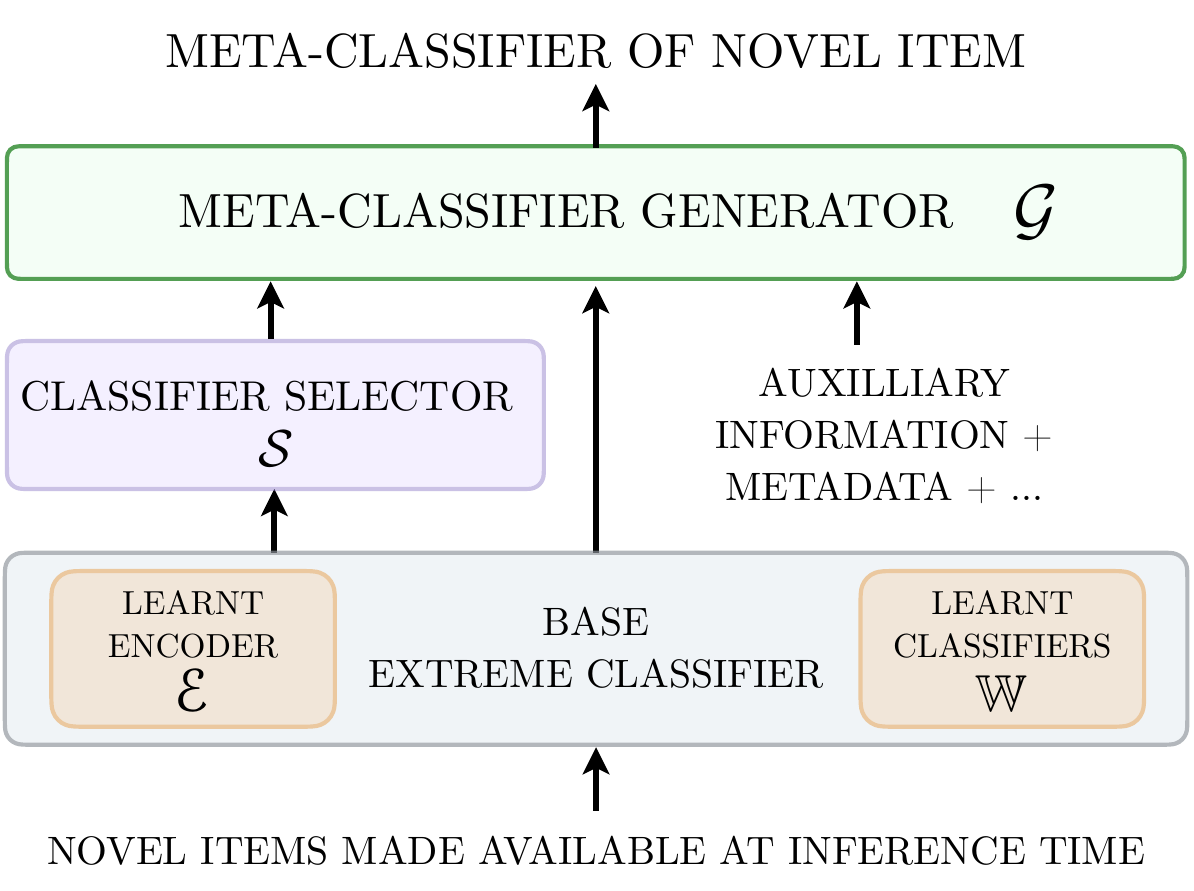}
\caption{An overview of our proposed \ul{E}xtre\ul{M}e \ul{MET}a-classifica\ul{T}ion (\xmec) framework. Given an extreme classification base ( encoder \(\gE\) and classifiers \(\sW\)), \xmec consists of two modules, where (a) the classifier selector \(\cS\) retrieves the \textit{most informative} classifiers for a \new item, and (b) the meta-classifier generator \(\cG\) combines the selected classifiers and other meta-data to create the meta-classifier. }
\label{Fig:EMMETT}
\end{figure}

Figure \textbf{\ref{Fig:EMMETT}} presents the architectural overview of \xmec. \xmec is a generic framework where different choices for the two modules and the training strategy can potentially yield algorithms with varying accuracy and efficiency trade-offs, which also presents ample opportunities for future research explorations. 

In this paper, we develop one such algorithmic instantiation of the \xmec framework, named \alg, for \ul{I}mproved \ul{RE}trieval of \ul{N}ovel it\ul{E}ms. Given a generic Siamese encoder augmented with extreme classifiers on observed items, \alg is designed to add meta-classifier functionality to this base with minimal effort. The \alg selects the classifiers associated with observed items by means of an efficient ANNS-search. The item representations are derived from the base Siamese encoder itself. This avoids training an additional module thus reducing the training and deployment efforts, but can also add noisy classifiers to the shortlist. To effectively filter the noise, combine the useful knowledge and synthesize meta-classifiers, \alg trains a separate meta-classifier generator model that is based on a transformer architecture. Given the shortlist of the selected classifiers, the meta-classifier synthesis is highly efficient, requiring a forward pass over the single-layer transformer (cf. Section~\ref{Sec:Ablations}). \alg training freezes the Siamese encoder and extreme classifiers, and only trains the generator using a weighted one-vs-all classification loss. Both novel item representation and retrieval are real-time and require few milliseconds in \alg. Figure \textbf{\ref{Fig:IRENE}} depicts the \alg architecture.

Theoretical analysis is a crucial tool which can guide the design of robust and effective machine learning systems. For this purpose, we develop a novel theoretical framework for analyzing the generalization performance in large-scale zero-shot retrieval. It is based on a key insight that, given a dense retriever that outputs calibrated query-item relevance scores, the zero-shot retrieval task is equivalent to a binary classification task with a query-item pair as our data sample. This equivalence allows us to leverage and further extend the theoretical ideas from traditional binary classification literature. Our results in Section \textbf{\ref{Sec:Theory}} show that (1) One-vs-All loss in \alg leads to strong zero-shot generalization performance by reducing the variance in the training loss, (2) Freezing extreme classifiers during \alg avoids overfitting and is therefore crucial for robust zero-shot generalization, (3) Number of observed items and the Generator's complexity also affect the zero-shot generalization. These results are also validated empirically in Section \textbf{\ref{Sec:Ablations}}.

In summary, this paper makes the following key contributions:
\begin{enumerate}[leftmargin=*]
\item A \textbf{novel and generic algorithmic framework}, \xmec, for learning accurate representations of novel items, thus improving the zero-shot retrieval performance.
\item A \textbf{novel and efficient algorithm}, \alg, that can be added to any Siamese encoder and is well-suited for large-scale real-world applications.
\item A \textbf{novel theoretical framework} for analyzing the generalization performance in large-scale zero-shot retrieval.
\item \textbf{Comprehensive experimental validation} demonstrating that \alg can offer up to $15\%$ gains in terms of \textit{Recall}@10 with minimal overheads, in zero-shot retrieval on diverse (varying in their scale and application) benchmark datasets.
\item \textbf{Online A/B tests} on a large-scale \textit{Sponsored Search} application in a major search engine which shows that \alg
can improve ad click rate by $4\%$ thereby demonstrating its real-world utility.
\end{enumerate}

The source code for \alg is available at \BBlue{\url{https://aka.ms/irene}}. The project page is \BBlue{\url{https://aka.ms/emmett}}.

%% file: sections/3_related_work.tex
\section{Related Work}
This section provides a comprehensive review of the past literature that is pertinent to our work. It presents the different classes of approaches available for retrieval and highlights their scalability, accuracy and zero-shot capabilities.

\noindent \textbf{Traditional statistical approaches:} Traditional approaches to large-scale retrieval, including TF-IDF~\citep{jones2021ASI} and BM25~\citep{robertson2009probframework}, relied on measuring the co-occurrences of terms in a query text and an item text. These methods utilized inverted term indices for efficient item search which scaled well and allowed rapid insertion of novel items. However, they lacked the ability to understand the context of the inputs. ZestXML~\citep{Gupta21}, a recent method for large-scale zero-shot retrieval, extended the term-matching to also account for semantic correlations between the terms which improved the retrieval performance. Nevertheless, these traditional methods have been largely superseded by semantic matching techniques using deep networks, which offer superior accuracy.

\noindent \textbf{Siamese-encoder approaches:} These approaches employ deep neural encoders to embed both query and item texts into a shared, low-dimensional representation space. Leading approaches such as DPR~\citep{Karpukhin20}, ANCE~\citep{Xiong20},  MACLR~\citep{Xiong21}, RocketQA~\citep{Qu21} and NGAME~\citep{Dahiya23}, utilize a transformer-based encoder. ANCE, RocketQA and NGAME also implement advanced negative-mining strategies for robust encoder training. MACLR leverages self-supervised learning based on an inverse cloze task for pre-training. In these approaches, incorporating a novel item is efficient and requires just a single encoder pass over their raw features. However, to maintain real-time response rates, these methods are constrained to use a small transformer, which degrades retrieval accuracy due to lack of world knowledge necessary for large-scale retrieval. Recent works like ColBERT~\citep{Khattab20} and Semsup-XC~\citep{Aggarwal23} modify the Siamese architecture to permit token-level interactions between the query and item texts which improves retrieval accuracy. However, this enhancement also increases inference latency, making these models less suitable for practical deployments. Semsup-XC also relies on web-scraped meta-data to boost zero-shot performance, which can be expensive.

\noindent \textbf{Extreme classification approaches}: These approaches learn one-vs-all classifiers to represent each observed item in the training set~\citep{Jain19,Zhang21,Mittal21b,Medini19, Saini21,Dahiya21,Dahiya23b,Kharbanda22,Gupta22}. These are highly accurate, offer low inference latencies and have been successfully applied to a wide range of large-scale retrieval tasks including document tagging~\citep{Babbar17}, product-to-product recommendation~\citep{Khandagale19,Jiang21}, and sponsored search~\citep{Dahiya21b,Dahiya23}. Unfortunately, most of these approaches do not have support for zero-shot items.

Among these, the most relevant to our work is DEXA~\citep{Dahiya23b}. DEXA learns a small set of classifiers, one classifier per cluster of items to improve the retrieval performance. DEXA can be viewed as a naive version of \xmec with an inferior architecture and training loss, and is significantly outperformed by \alg (see Section \textbf{\ref{Sec:Experiments}}).

\noindent \textbf{Other zero-shot retrieval approaches:} Several approaches have been proposed for zero-shot multi-label retrieval such as ESZSL~\citep{Bernardino15}, COSTA~\citep{Mensink14}, LAF~\citep{Liu20} \textit{etc.} These works consider a few thousand labels (items) and typically do not scale to million items.

Over the past year, large language model based zero-shot retrieval is also gaining popularity~\citep{Shen23}. These models are expensive and are not yet demonstrated for large-scale retrieval. 

The works in ~\citep{Gao23,Xin22} study a different zero-shot setting where the task itself is novel with no available relevance signals for model training. They propose techniques based on domain transfer as solutions. Unlike these works, our setting assumes that abundant click data is available as relevance signals and focusses on leveraging them to improve the retrieval with zero-shot items.
 

%% file: sections/4_method.tex
\section{Extreme Meta-classification} \label{Sec:EXMEC}

In this section, we introduce {\bf \xmec}, our proposed \ul{E}xtre\ul{M}e \ul{MET}a-classifica\ul{T}ion framework for improving zero-shot retrieval by means of incorporating the knowledge present in the learnt classifiers of observed items. \par

\subsection{Preliminaries} \label{SubSec:Prelims}
We describe the notation and background that is necessary for the rest of this section. Consider queries \(X\) drawn from the space \(\gX\), which are mapped to items \(Z\) drawn from the item space \(\gZ\). During training, we have access to \(N\) data points \(\{X_i\}_{i=1}^N\) and \(L\) \it{ observed} items  \(\{Z_\ell\}_{\ell=1}^L\). Each data-item pair is associated with a label \( y_{i\ell} = \{0,1\}\), which is 1 when the item is relevant to the input query, and 0 when the item is irrelevant to the query. \par
\noindent {\bf Dense retrieval} (DR) algorithms consider a deep feature encoder \(\gE\) that projects the data samples and items to the space of \(d\)-dimensional encoder representations, denoted as \(\vx = \gE(X)\) and \(\vz = \gE(Z)\), respectively, both \(\gE(\gZ),\gE(\gX)\) belonging to \(\sR^d\). The query and observed-item sets in this encoder space are represented as \(\sX = \{\vx_i\}_{i=1}^N\), and \(\sZ = \{\vz_{\ell}\}_{\ell=1}^L\), respectively. In Siamese DR algorithms, both query and item representations share the same encoder space. Items are retrieved for a given query, typically using approximate nearest neighbor search (ANNS), based on the similarity score 
 \(\vx_i^{\top}\vz_\ell\). \par

\noindent {\bf Extreme classification} (XC) algorithms are built on the intuition that classifier-based item representations can surpass the performance of encoder-representations. In XC, each item \(\vz_{\ell}\) is represented by a \ova classifier \(\vw_\ell\), collectively denoted as \( \sW = \{\vw_\ell\}_{\ell=1}^{L}\). These are trained using triplet-based~\citep{Dahiya23} or cross-entropy-based~\citep{renee_2023} losses atop the encoder representation. The similarity score is calculated between the query representation and the classifiers \((\vx_i^{\top}\vw_\ell)\).  Thus, the XC module comprises an encoder and a set of learned classifiers \((\gE,\sW)\). While effective in retrieving observed items, these methods face challenges in retrieving \new items.

\subsection{The \xmec Framework} \label{SubSec:EMMETT}

We propose \xmec, an {\bf extreme meta-classification} framework for zero-shot generalization in the retrieval setting. The \xmec framework is designed to develop highly accurate and efficient retrieval models capable of handling \new item retrieval during inference. Existing XC models, despite their accuracy through classifier-based approaches built atop encoder representations, suffer from high data and resource requirement and therefore do not generalize to the zero-shot setting.  \

The \xmec framework considers training data \(\sX\) and \(\sZ\), a {\bf base extreme classifier \((\gE,\sW)\)}, comprising the encoder \(\gE\) and observed-item classifiers \(\sW\), and a set of \new items \( \sZ_n = \{\vz_{\ell}\}_{\ell=1}^{L_n}\) made available at inference time, \xmec includes two modules:
\begin{itemize}[leftmargin=*]
    \item {\bf Classifier Selector \(\cS\)}: This module selects the \it{ most informative} set of classifiers from the base XC model for a \new item \(\vz \in \sZ_n\) introduced at inference time.
    \item {\bf Meta-classifier Generator \(\cG\)}: This module combines the  classifiers selected by \(\cS\), and other auxiliary information (such as the encoder representation \(\vz\)) to create the meta classifier \(\vu\) associated with item \(\vz\). 
\end{itemize}

\xmec offers flexibility in its implementation through a variety of design choices, such as the base XC framework, choice of encoder architecture and training methods. The classifier selector's \(\cS\) design involves choosing a selection algorithm and determining the number of classifiers to select. An ideal selection algorithm should choose \it{ informative} classifiers, taking into account factors such as named entities, numbers, ambiguous phrases, synonyms, etc. Since the task of classifier selection can be viewed as retrieval, any existing retrieval methods can be used. In the context of \xmec we prioritize efficient classifier selection and quick incorporation of \new items. 
The generator's \(\cG\) design must take into consideration model complexity (which we discuss in Section~\ref{Sec:Theory}), ensuring it effectively incorporates the selected classifiers and potential auxiliary data to learn the meta-classifier \(\vu\) without overfitting. Another component of \xmec is in designing loss functions towards zero-shot generalization, and choosing appropriate training strategies, such as end-to-end, modular, etc. \par 

Figure~\ref{Fig:EMMETT} provides a visual illustration of these three components. \xmec enables interpreting existing XC models' zero-shot generalizability. For example, in DEXA~\citep{Dahiya23b}, given a \new item, its classifier can be selected based on cluster assignment, and then summed with its encoder representation. Similarly, in NGAME~\citep{Dahiya23}, the classifier selector is an indicator function, which returns a classifier only for the observed items. NGAME's decision tree, along with its encoder representation and classifier, can be viewed as its generator block. DEXA's and NGAME's relatively poorer zero-shot generalization can be attributed to their simplistic \(\cS\) and \(\cG\) architectures. An additional reason why these models perform poorly on \new items can be traced back to their loss functions, which are not tailored for zero-shot performance. \par
In Section~\ref{Sec:IRENE}, we present \alg, a specific instance of \xmec designed for zero-shot generalization. This includes a thoughtful design of classifier selector and generator modules, considering training strategies, deployment ease, new item incorporation at inference, and loss functions prioritizing zero-shot performance. Theoretical analysis of \alg's effectiveness is detailed in Section~\ref{Sec:Theory}.

%% file: sections/5_implementation.tex
\section{The IRENE Extreme Meta-Classifier} \label{Sec:IRENE}
We now present the {\bf \alg} algorithm, our proposed approach for \ul{I}mproved \ul{RE}trieval of \ul{N}ovel it\ul{E}ms with extreme meta-classification. \alg is adaptable to any XC framework. Specifically, we utilize a 6-layer DistilBERT encoder, trained with state-of-the-art, computationally efficient algorithms such as NGAME~\citep{Dahiya23}, ANCE~\citep{Xiong20}, MACLR~\citep{Xiong21}, and DPR~\citep{Karpukhin20}. Given a real-world application such as product-to-product recommendations, document tagging, or matching user queries to advertiser keywords appropriate feature encoders may be chosen. After training, the encoder \(\gE\) and the observed-item classifiers \(\sW\) remain fixed through subsequent stages. The learnt classifiers are accessible via a lookup function \(\cC_{lf}\) based on the items' encoder representations, \it{i.e.,} \(\vw_\ell = \cC_{lf}(\vz_\ell)\). \par

\noindent {\bf The classifier selector} \(\cS\), for any item \(\vz_\ell\), either observed or \new, retrieves \(K\) classifiers of related items using an ANNS index built on the item representations \(\gE(\sZ_s)\), served via approaches such as DiskANN~\citep{jayaram2019diskann}. The selection is based on a maximum inner product search~(MIPS) between the \new item, and \seen items $\argmax_{K} \{\vz_\ell^{\top} \vz_o;\,\forall\,\vz_o \in \sZ\}$. The complexity of these approaches has typically been shown to be logarithmic in the number of items, \it{i.e.,} $\mathcal{O}(\log L)$. The choice of \(K\), a key hyper-parameter, balances model capacity and complexity, which we formalize in Section~\ref{Sec:Theory} and validate in Section~\ref{Sec:Ablations}. Empirically, we found \(K\approx3\) to work well. 

\begin{figure*}[!th]
\centering
\includegraphics[width=\linewidth]{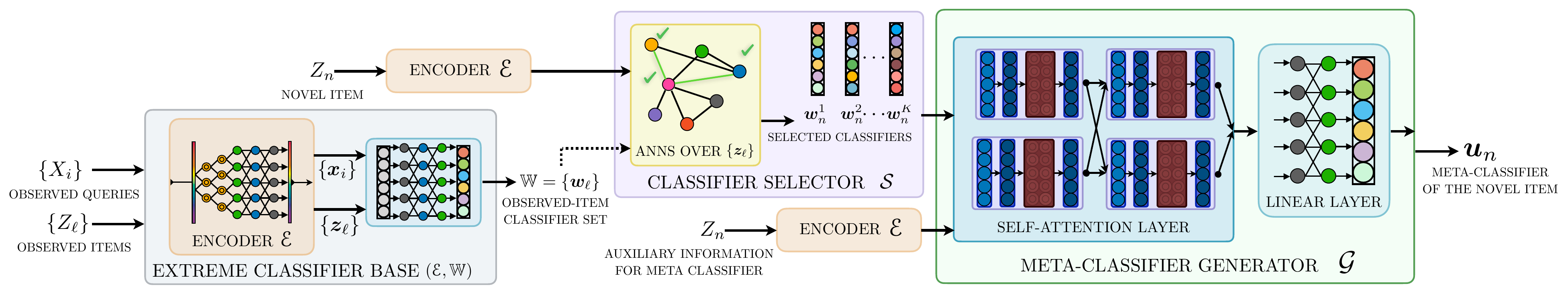} \\[-10pt]
\caption{The \alg extreme meta-classifier. IRENE comprises (a) A base extreme classifier encoder and the classifiers trained via a standard algorithm; (b) A classifier selector \(\cS\), which, given the encoder representation of a \new item \(Z_n\), retrieves \(K\) classifiers based on an approximate nearest neighbor search (ANNS), and (c) A transformer-based meta-classifier generator \(\cG\). The meta-classifier \(\vu_n\) of \(Z_n\) is computed as given in \Eqref{Eqn:Generator}. }
\label{Fig:IRENE}

\end{figure*}

\noindent {\bf The meta-classifier generator  $\cG_{\vphi}$} is based on a transformer architecture, inspired by their capability to be universal approximators~\citep{Yun2020Are} and their success in learning embeddings in few-shot scenarios~\citep{FewShotEmbed2020}. While the transformer layers are capable of handling a variety of input embeddings, in \alg, the input sequence consists of the selected classifiers and the \new item's encoder representation and is passed through self-attention and linear blocks.  In particular, the meta classifier is given by \(\vu_\ell =\cG_{\vphi}(\vz_\ell,\cS(\vz_\ell))\)  and each layer can be expressed as:
\begin{align}
    \text{Linear}(\text{Self-Attention}(\vz_\ell + \vt_{enc}, \vw_{\ell}^{1} + \vt_{clf}, \vw_{\ell}^{2} + \vt_{clf}, \ldots),
    \label{Eqn:Generator}
\end{align}
where \(\vw_{\ell} \in \cS(\vz_\ell)\), $t_{enc}$ and $t_{clf}$ are learnable \it{type embeddings} that distinguish between encoder and classifier inputs. The transformer layer is capable of learning correlations between the classifiers and the encoder representation, yields superior meta-classifiers over methods that consider naive weighted summation for the generator (cf. Section~\ref{Sec:Ablations}). The computational complexity is \(\bigO{\left(K+1\right)^2d^2}\) for the self attention layer~\citep{Vaswani17Attention}, and \(\bigO{\left(K+1\right)d}\) for the linear layer. Figure~\ref{Fig:IRENE} illustrates the modules present in the \alg framework.

\subsection{\alg: Training and Inference} \label{SubSec:Traning}

In this section, we outline the training strategy for \alg. Firstly, we {\bf assume that the representation encoder \(\gE\), and the set of observed-item classifiers \(\sW = \{\vw_{\ell}\}_{\ell=1}^L\) are made available to us a priori, and remain fixed.}  With these components \((\gE,\sW)\), the ANNS index is constructed over the encoder representations of the observed items \(\sZ\). Classifier selection is as described in Section~\ref{Sec:IRENE}.

\alg's meta-classifier training employs a binary cross-entropy loss reformulated for zero-shot performance. In particular, to simulate \new items at training time, given an observed item \(\vz_\ell\), instead of using the item's own classifier \(\vw_\ell\), the meta-classifier \(\vu_\ell\), derived from the selected classifiers, is utilized. That is, an observed item is trained with its selected \it{neighboring} classifiers and not its own. The original targets \(y_{i\ell}\) remain unchanged. The loss \(\gL_{\gG}\) is:
\begin{align}
    -\sum_{\vx_i\in\sX}\sum_{\vz_\ell\in\sZ} \left(C y_{i\ell}\ln\left(\sigma(\vx_i^{\top}\vu_\ell)\right)\!+\!(1\!-\!y_{i\ell})\ln\left(1-\sigma(\vx_i^{\top}\vu_\ell)\right)\!\right)
\end{align}
where $\sigma(\cdot)$ is the sigmoid function, \(\vu_{\ell} =\cG_{\vphi}(\vz_\ell,\cS(\vz_\ell)) \), \(\vx = \gE(X)\), and \(C\) is the weight for misclassified positive query-item pairs. This weight addresses the imbalance between positive and negative query-item pairs in standard XC settings, as misclassifying positives significantly impacts performance~\citep{buvanesh2024enhancing}. We link the likelihood of encountering positive-pairs to zero-shot generalization performance in Section~\ref{Sec:Theory}. \par

To facilitate training, the loss is approximated using standard negative mining strategies, with the positive, and negative-mined items associated with a query represented by \(\sZ_{o,i}^+\) and \(\sZ_{o,i}^-\), respectively. The modified loss \(\gL_{\gG}\) is:
\begin{align}
    -\sum_{\vx_i\in\sX}\sum_{\vz_\ell\in\sZ_{o,i}}\!\!\!\left(C y_{i\ell}\ln\left(\sigma(\vx_i^{\top}\vu_\ell)\right)\!+\!(1\!-\!y_{i\ell})\ln\left(1\!-\!\sigma(\vx_i^{\top}\vu_\ell)\!\right)\!\right)\!,
\end{align}
where \(\sZ_{o,i} = \sZ_{o,i}^+ \cup \sZ_{o,i}^-\) includes both positive items and mined negative items associated with query \(\vx_i\).

\noindent {\bf Inference}: An Approximate Nearest Neighbor Search (ANNS) index $\cA$ is constructed over the item representations to facilitate zero-shot inference. This index encompasses either solely \new item representations for zero-shot inference or a combination of observed and \new items for generalized zero-shot inference. Thus, $\cA$ is built over $ \{\vz_\ell~|~\vz_\ell \in \sZ_T\}$, where \(\sZ_T\) represents the set of items at inference, either \(\sZ_n\) or \(\sZ\cup\sZ_n\).

During inference, the embedding of a query \(T\) is calculated as $\vx_T = \gE(X_T) \in \sR^d$ with $X_T$ being the query text.  To retrieve relevant items for $\vx_T$, $\cA$ is queried using $\vx_T$. The computational complexity of \alg's inference is $\Omega\left(E + d \log\left({|\sZ| + |\sZ_n|}\right)\right)$, where $E$ is the cost of encoder \(\gE\). \alg is designed to seamlessly incorporate new items by calculating their representations and adding them to the ANNS index~\citep{singh2021freshdiskann} as and when \new they arrive in the system. This feature is particularly advantageous for applications like sponsored search, where \new items are frequently introduced. In such scenarios, re-training the model with a large volume of items can be resource-intensive, making \alg a viable and efficient solution.

%% file: sections/6_theory.tex
\section{Theoretical Guarantees} \label{Sec:Theory}
Prior to evaluating the experimental efficacy of \alg, we establish theoretical guarantees for the zero-shot generalization of \xmec, with a specific focus on the case of \alg. To simplify our analysis, we represent data samples as query-item pairs, \(\vs = (\vx,\vz)\). The dataset \(\sS = \{\vs\}\) is formed through two potential approaches. The most straightforward method involves randomly selecting \(N\) queries and items from the respective distributions \(\cX\) and \(\cZ\), pairing them to create \(\tilde{\sS} = \{\vs_i = (\vx_i,\vz_i)~|~i=1,2,\ldots N\}\). Alternatively, considering all possible query-item pairings from \(\sX\) and \(\sZ\) (as defined in Section~\ref{SubSec:Prelims}), results in \(\sS = \{ (\vx_i,\vz_\ell)~|~i=1,2,\ldots N, \ell = 1, 2,\ldots L\}\), meaning, \(\vs\) is drawn from the Cartesian product space \(\gX\times\gZ\). The dataset is indexed by \( j = L(i-1) + \ell,~j=1,2,\ldots M = NL\). In this context, the extreme (meta) classification problem becomes binary, with target values \(y\) in \(\gY\), where \(y_{j} = 1\) indicates a positive association between item \(\vz_{\ell}\) and the query \(\vx_i\). Within the XC setting, it is well known that negatively associated pairs are significantly more likely to occur than positive ones. We assume that any set \(\sS\) contains at most \(\kappa\) positively associated pairs.
We also assume norm bounds on the encoder representations of the queries, and the learnt classifiers, \it{i.e.,} \(\max_{\vx\in\gE(\gX)} \|\vx\|_2 \leq B \) and \(\max_{\vw\in\sW} \|\vw\|_2 \leq W \), respectively. \par

In assessing the generalization performance of \xmec, we recall the concepts of empirical and true risks. Given a function \(f(\cdot)\), and a chosen loss function, the empirical risk over a dataset \(\sS\), and the true risk are defined as:
\begin{align*}
    \hat{\cR} = \frac{1}{M} \sum_{\subalign{j=1\\\vs_j\sim\sS}}^{M} \mathrm{loss}(f(\vs_j),y_j)~\text{and}~\cR = \mathbb{E}_{\vs\sim\gX\times\gZ}\left[ \mathrm{loss}(f(\vs),y) \right],
\end{align*}
respectively. For our analysis, we use the weighted binary cross-entropy loss, \(Cy\ln(f) + (1-y)\ln(1-f).\) We assume that the function \(f\) belongs to the class \(\cF\). \par
The generalizability of a model is indicated by the difference between empirical and true risks. This deviation depends on both data and the model. The model's influence is typically quantified using the \it{Rademacher complexity} \(\mathfrak{R}\)~\citep{FML12} of the function class \(\cF\). This complexity can be evaluated empirically over the dataset (known as the empirical Rademacher complexity \(\hat{\mathfrak{R}}_{\sS}(\cF)\)), or across all datasets of size \(M\), given by \(\mathfrak{R}_M(\cF) = \mathbb{E}_{\sS}\left[\hat{\mathfrak{R}}_{\sS}(\cF)\right]\). Data dependence is often highlighted by bounding the risk deviation using the McDiarmid inequality~\citep{FML12}. The following theorem bounds the deviation of the true risk from the empirical risk.

\begin{theorem}\label{Theorem:GenBound}
    {\bf (Generalization performance of \xmec)} Let \(R\) and \(\hat{R}\) denote the true and empirical risk, respectively, and \(\hat{\mathfrak{R}}_{\sS}\) denote the empirical Rademacher complexity over set \(\sS\). Let \(p \ll 1\) be the probability that a query-item pair is positively associated (\it{i.e.,} \(y_j=1\)), and \(q\) denote the probability that a set \(\sS,~|\sS| = M\) has at most \(\kappa\) positive pairs. Then, with probability at least \(1-\delta\), \(\delta \in (2q,1)\), we have the following generalization bound:
    \begin{align}
        R &\leq \hat{R} + \hat{\mathfrak{R}}_{\sS}(\cF) + 3\left(q + \sqrt{\frac{\ln\left(\frac{2}{\delta-2q}\right)}{2M}}\right), ~\text{where} \\
        q &\leq \exp\left\{-2M\left(1-p-\frac{\kappa}{M}\right)^2\right\}.
    \end{align}
\end{theorem}
\begin{proof}
    The detailed proof is provided in Appendix~\ref{App:Proofs}. We summarize the proof here. Without loss of generality, we redefine the loss to have the target labels \(y_{j} \in \{-1,1\}\), giving rise to the following form of the loss:
\begin{align*}
    g(\vs,y) = \mathrm{loss}\left(f(\vs),y\right) = \left(\frac{1-y}{2}\right)f(\vs) - C \left(\frac{1+y}{2}\right)f(\vs).
\end{align*}
We then define the function \(\Phi(\sS)\) as follows
\begin{align*}
    \Phi(\sS) = \sup_{f\in\cF} \left\{ \bbE\left[ g(\vs,y)\right]- \hat{\bbE}\left[ g(\vs_j,y_j)\right] \right\}.
\end{align*}
By means of an extended McDiarmid inequality (that accounts for the  positive-negative sample imbalances)~\citep{FML12}, we get
\begin{align*}
\mathrm{Pr}\left( \left|\Phi(\sS) - \mathbb{E}_{\sS}\left[\Phi(\sS)\right]\right| \geq \epsilon \right) &\leq 2q + 2\exp \left\{ - 2M\left(\epsilon - q\right)^2\right\} = \delta \\
\Rightarrow \epsilon &= q + \sqrt{\frac{\ln\left(\frac{2}{\delta-2q}\right)}{2M}}.
\end{align*}
Given the inequality \(\mathbb{E}_{\sS}\left[\Phi(\sS)\right] \leq 2 \mathfrak{R}_M(\mathrm{loss} \circ\cF) = 2\hat{\mathfrak{R}}_{\sS}(\cF)\)~\citep{FML12}, substituting for \(\delta\) into the first equation, and simplifying, we get 
\begin{align*}
    R &\leq \hat{R} + \hat{\mathfrak{R}}_{\sS}(\cF) + 3q + 3\sqrt{\frac{\ln\left(\frac{2}{\delta-2q}\right)}{2M}}.
\end{align*}
To bound \(q\), we bound the probability that the set \(\sS\) contains at least \(M-\kappa\) positively associated pairs by means of the Hoeffding inequality, which completes the proof of Theorem~\ref{Theorem:GenBound}.
\end{proof}
Theorem~\ref{Theorem:GenBound} sheds light on the generalizability of new items in XC meta classifiers. Consistent with standard findings, we note that generalization gap is inversely related to the dataset size, \(M\). Furthermore, the choice of dataset \(\sS\) over \(\tilde{\sS}\) reveals distinct advantages when considering the probabilities \(p\) and \(q\). Specifically, a \ova setting, wherein we consider the Cartesian product space of queries and items, will lead to better overall performance. Further, for sufficiently small \(\kappa\) and sufficiently large \(M\), we have \(q \approx \exp\left\{-2M\right\}\). We also observe a direct correlation between generalization, and the complexity of the function class \(\cF\). In scenarios where \(\cF\) corresponds to the class of meta-classifier generators, the associated Rademacher complexity is detailed in the following lemma:
\begin{lemma}
\label{Lemma:IRENE}
({\bf Rademacher complexity of the \alg generator}) Let \(\cF\) be the class of functions defined in the \alg algorithm, comprising pre-determined encoder representations and classifiers, a given classifier selector that outputs \(K\) classifiers, and \(\cG\), the meta-classifier generator.  Then, the Rademacher complexity of \(\cF\) can be bounded as follows:
\begin{align}
    \hat{\mathfrak{R}}_{\sS}(\cF) \leq \bigO{B \|\vM\|_2 \sqrt{d \ln(K+1)}},
\end{align}
where \(\vx\in\sR^d\) and \(\vM \in \sR^{d\times1}\) is the weight matrix associated with the linear layer.
\end{lemma}
\begin{proof}
    The detailed proof is provided in Appendix~\ref{App:Proofs} and follows by repeated application of Talagrand’s lemma~\citep{FML12}:
\begin{align*}
\hat{\mathfrak{R}}_{\cS}(\cF) &\leq B~ \hat{\mathfrak{R}}_{\cS}(\cF_1)\\
&\leq B \|\vM\|_2~\hat{\mathfrak{R}}_{\cS}(\cF_2) \\
&\leq B \|\vM\|_2 \sqrt{d\ln\left(K+1\right)}~\hat{\mathfrak{R}}_{\cS}(\cF_3)\\
&\leq \bigO{B \|\vM\|_2 \sqrt{d\ln\left(K+1\right)}},
\end{align*}
where \(\cF_3 = \cC_{lf} \left( \cS(\vz_\ell)\right)\), \(\cF_2 = \mathrm{Self Attn.}(\cF_3)\) and \(\cF_1 = \vM \cF_2 + \vb\).
\end{proof} 
Lemma~\ref{Lemma:IRENE} yields several critical insights. Firstly, for a given XC base (consisting of an encoder and classifiers), the complexity of the meta-classifier increases logarithmically with \(K\). This creates a balancing act: using a single classifier may result in inadequate model capacity, but increasing \(K\) increases model complexity, potentially degrading performance, and slowing down training. Ablation experiments presented in Section~\ref{Sec:Ablations} validate this claim. Secondly, the complexity of the meta-classifier generator is independent of \(L\), indicating that the model effectively scales in generalizing to \new labels. This scalability is attributed to the classifiers being fixed during the generator's training. The subsequent corollary discusses the Rademacher complexity of the function class \(\cF\) when the XC classifiers and the generator \(\cG\) are updated concurrently.
\begin{corollary}
\label{Corollary:IRENE}
({\bf Rademacher complexity of the \alg generator with trainable classifier}) Let \(\cF\) be the class of functions defined in the \alg algorithm as in Lemma~\ref{Lemma:IRENE}. Let the classifier set \(\sW\) be trainable over the meta-classifier loss. Then, the Rademacher complexity of \(\cF\) can be bounded as follows:
\begin{align}
    \hat{\mathfrak{R}}_{\sS}(\cF) \leq \bigO{B^2W \sqrt{\frac{L}{M}} \|\vM\|_2 \sqrt{d \ln(K+1)}}.
\end{align}
\end{corollary}
This result is obtained by combining Lemma~\ref{Lemma:IRENE} with an extension of Theorem~3 present in~\citet{awasthi20adv}. Training classifiers via the meta-classifier loss incurs a complexity in the order of \(\sqrt{L}\). For the sake of completeness, we present the extension of Theorem~3 present in~\citet{awasthi20adv}, relevant to the XC setting, with the model trained on \(N\) data points \(\tilde{\sX}\), and the loss defined over classifiers \(\sW\). \par 
The following Lemma bounds the Rademacher complexity of the XC classifier class:
\begin{lemma}
\label{Lemma:XC}
({\bf Rademacher complexity of the XC classifiers}) (extension of~\citet{awasthi20adv}, Theorem~3) Let \(\cF\) be the class of linear classifiers defined over the seen-item set \(\sZ_s\) in the classical XC setting (cf. Section~\ref{SubSec:Prelims}), \it{i.e.,} \(\cF = \{\langle \vx,\vw_{\ell}\rangle~|~\ell=1,2,\ldots,L\}\). Then, the Rademacher complexity of \(\cF\) can be bounded as
\(
    \hat{\mathfrak{R}}_{\sS}(\cF) \leq \frac{LBW}{\sqrt{N}},
\)
where \(|\sX| = N\) is the Cardinality of the training set.
\end{lemma}
The proof is provided in Appendix~\ref{App:Proofs}. As expected, the complexity of the XC classifier, particularly in terms of generalizing to unseen labels, increases linearly with the number of labels. This is because the classifier's learning depends on the query-item pairs available during training. Therefore, as \new items arrive, additional data is necessary for accurately learning the classifiers. \par

These findings underscore the effectiveness of the meta-classifier-based \xmec framework in achieving zero-shot generalization. They also validate the design choices implemented in the \alg generator and algorithm, which contribute to a model with favorable model complexity. We now proceed to provide experimental evidence to support the \alg algorithm, complemented by ablation studies. These studies elucidate our design decisions, connecting them to the theoretical bounds established earlier.

%% file: sections/7_experiments.tex
\section{Experimental Results}
\label{Sec:Experiments}

\textbf{Datasets:}  We validate the performance of \alg across a diverse set of datasets spanning multiple applications, and label space sizes (cf. Table~\ref{tab:datasets_applications}). For LF-AmazonTitles-1.3M and LF-Wikipedia-500K, we use the experimental setup reported in~\citet{Bhatia16}, while for LF-AOL-270K and LF-WikiHierarchy-550K, we refer to~\citep{buvanesh2024enhancing}. We create the zero-shot versions of these datasets by randomly partitioning the set of items into \seen and \unseen using a 90-10 split ratio~\citep{simig-etal-2022-open}. 

\begin{table}[!t]
\fontsize{8}{10}\selectfont
\centering
\caption{A summary of the datasets, and their corresponding applications, used in evaluating \alg.}
\begin{tabular}{c|c|c}
\toprule
\textbf{Dataset} & \textbf{Application Type} & \textbf{Feature Type} \\
\midrule
LF-AOL-270K & Query Completion & Short-Text \\
LF-Wikipedia-500K & Category Annotation & Long-Text\\
LF-WikiHierarchy-550K & Taxonomy Completion & Short-Text \\
LF-AmazonTitles-1.3M & Product Recommendation & Short-Text \\
KeywordPrediction-10M & Sponsored Search & Short-Text \\
\bottomrule
\end{tabular}
\label{tab:datasets_applications}
\end{table}

\noindent \textbf{Baselines:} We demonstrate \alg's efficacy by comparing its performance when built atop leading dense retrieval encoders such as NGAME~\citep{Dahiya23}, ANCE~\citep{Xiong20}, MACLR~\citep{Xiong21}, and DPR~\citep{Karpukhin20}. We also compare against competitive zero-shot XC methods such as SemSup-XC~\citep{Aggarwal23}, ZestXML~\citep{Gupta21}.

\noindent \textbf{\alg's Training Procedure}: Given a base encoder $\gE$, the classifiers $\sW$ are trained using a one-versus-all BCE loss formulation similar to Ren\'ee~\citep{renee_2023}. However, unlike Ren\'ee, where the encoder and classifiers are trained jointly, we train only a single output-side transformer layer jointly with the classifiers. Subsequently, \alg's meta-classifier generator block ($\cG$) is trained using the hyper-parameters $K=3$ and $D=1$ for results reported in Table \ref{tab:main-zs-gzs-acc}. Ablations on the choice of \(K\) and \(D\) are provided in Section~\ref{Sec:Ablations}.\par
Additional details on dataset creation and statistics are provided in Appendix~\ref{sec:supp_dataset_creation}.

\noindent \textbf{Evaluation Setting and Metrics}:
Following ~\citep{Aggarwal23,Gupta21} we consider two settings: (i): zero-shot retrieval on \unseen items and (ii): generalized zero-shot retrieval on the combined set of \seen and \unseen items. We report performance on metrics such as  Precision@$k$ (P@$k$), and Recall@$k$ (R@$k$) at different truncation levels $k$.

\input{tables/main_zero_shot_gen_zero_shot}

\noindent \textbf{Results on Benchmark Datasets:} 
Table \ref{tab:main-zs-gzs-acc} presents comparisons of the zero-shot and generalized zero-shot performance of \alg against the baselines. When added atop different encoders, \alg consistently improves the performance in both the zero-shot (+1.5 to +40\% in P@1 and +1 to +45\% in R@10) and generalized evaluation settings (+1 to +29\% in P@1 and +1 to +15\% in R@10). \alg shows larger gains on datasets such as LF-AOL-270K and LF-WikiHierarchy-550K where queries and their associated items have higher lexical dissimilarity. While modelling such relations can be challenging for small encoders, they nevertheless function as effective neighbour selectors, thereby identifying relevant classifiers for combining via \alg's meta classifier generator. This is particularly visible in the case of MACLR, where \alg improves P@1 by nearly 40\% in zero-shot on LF-WikiHierarchy-550K. \alg's consistent gains over a wide range of applications on different encoders show the value of the information introduced by classifiers and emphasize the versatility of our approach. \par When compared to zero-shot XC methods such as SemSup-XC, \alg attains higher P@1 in the zero-shot setting (+10.31\%) while being approximately 350 times more efficient than SemSup-XC at inference (cf. Table~\ref{tab:inference_breakdown}). These gains across various baselines establish the modularity of \alg which can be used as a plug-and-play module atop any dense retriever to obtain more accurate representation for \new items for diverse recommendation applications.

\begin{table}[!t]
    \centering
    \fontsize{9}{9}\selectfont
    \caption{Comparison of inference time (in ms) on a single V100 GPU for the LF-AmazonTitles-1.3M-10 dataset. This includes (i) generating the representation of an unseen item (Representation Time/ Rep. Time), and (ii) retrieving relevant items for a given query (Retrieval Time). Compared to dual encoder methods like NGAME and DEXA, \alg adds no latency overhead, with the representation time increasing by 0.4 milliseconds. However, IRENE is about 350 times faster than SemSup-XC, which aggregate token-level similarities between queries and documents to obtain the final scores.}
    \label{tab:inference_breakdown}
    \vskip-7pt
    \begin{tabular}{@{}c|cc@{}}
        \toprule
        \multirow{2}{*}{\textbf{Method}} & \textbf{Rep. Time} (ms) \(\downarrow\) & \textbf{Retrieval Time} (ms) \(\downarrow\) \\
        &  (per item)  & (per query)  \\
        \midrule
        NGAME       & 0.08 & 0.43   \\
        SemSup-XC   & N/A   & 151.51  \\
        DEXA        & 0.48    & 0.43   \\
        NGAME + \alg & 0.54 & 0.43 \\
        \bottomrule
    \end{tabular}
    \vskip-10pt
\end{table}

\noindent \textbf{Case-study on Sponsored Search:}
In the application of sponsored search, accurately matching user queries to billions of advertiser bid keywords presents a formidable challenge, exacerbated by varying bid amounts based on relevance and semantic relationship of the match. We demonstrate \alg's efficacy by conducting offline experiments and online A/B tests on live search engine traffic. \alg trained on the KeywordPrediction-10M dataset was used to obtain representations for 100M novel keywords which were found to be 4\% more accurate than those obtained from leading dense retrievers in production, when evaluated in terms of the R@100 metric. In live deployment, \alg was found to increase the click-through rate (ad clicks obtained per unit query) and decrease the quick-back rate (fraction of users who quickly closed the ad) by 4.2\% and 0.9\%, respectively. In \alg a novel keyword can be encoded in under 1ms and the approach demonstrates a 13\% increase in good keyword predictions, as ascertained by expert judges. To handle potential distribution shifts that could occur in such dynamic settings, we can continually grow the base extreme classifier set by adding new item classifiers into it, as and when they receive clicks. Please refer to Appendix~\ref{sec:supp_sponsored_search} for detailed studies on \alg's application to sponsored search.

\noindent \textbf{\alg's Computational Cost:}
 Table~\ref{tab:inference_breakdown} presents comparisons on the time taken for generating representations for a \new item and retrieving the relevant items given a query. When incorporating a \new item, NGAME performs a forward pass over a 6-layer Distil-BERT. DEXA incurs an additional cost of an ANNS search over the cluster centroids. \alg, on the other hand, incurs additional costs from the classifier selector \(\cS\) and meta-classifier generator \(\cG\), which involves an ANNS search over the set of \seen items and a forward pass through a one-layer encoder respectively. The representation generation step was executed on a single NVIDIA Tesla V100 GPU for each method, while the ANNS search was carried out on a 96-core CPU machine.

While being efficient at inference, \alg's classifier selector (\(\cS\)) and meta-classifier generator (\(\cG\)) incur minimal training costs. On a single NVIDIA Tesla V100 GPU, the training time for the NGAME encoder on the LF-AmazonTitles-1.3M-10 dataset was 83 hours and training \alg on top of the NGAME encoder took only 6 hours.

\section{Ablations}
\label{Sec:Ablations}

We perform ablation experiments to evaluate the impact of various design choices within the components of the meta-classifier generator (\(\mathcal{G}\)) and the classifier selector (\(\mathcal{S}\)) in \alg. 

\noindent \textbf{Meta-classifier Generator \(\cG\):} Table~\ref{tab:combined-ablations} shows the effect of increasing the number of layers in \alg's meta-classifier generator block (\(\cG\)). Increasing the number of layers from 1 to 2 improves zero-shot P@1 by about 1\%. However, in increasing \(\cG\)'s depth to 4, the performance plateaus, which could be attributed to overfitting associated with increased complexity of \(\cG\). Additionally, we try out simpler alternatives, namely sum and weighted sum. \alg performs better than these simpler alternatives by 24\% P@1 in zero-shot evaluation underscoring the significance of \alg's attention-based meta-classifier generator.\par

\noindent \textbf{Classifier Selector $\cS$:} We evaluate the effect of changing $K$, the number of \seen items retrieved by $\cS$, given an ANNS-based \(\cS\), on zero-shot performance. From the results presented in Table~\ref{tab:combined-ablations}, we observe that increasing $K$ beyond 3 causes the performance to initially plateau. Subsequently, increasing $K$ to 20 brings about a decrease in performance. This is consistent with observations made in Section~\ref{Sec:Theory}, wherein smaller \(K\) yield a tighter generalization bound (and therefore, superior performance), as derived in Lemma~\ref{Lemma:IRENE}. \par

\begin{table}[!t]
\centering
\fontsize{9}{12}\selectfont
\caption{Ablation study on meta-classifier generator $\cG$ and classifier selector $\cS$ component in NGAME + \alg on LF-WikiHierarchy-550K-10 dataset for zero-shot evaluation. The depth \(\mathrm{D}\) denotes the number of layers in the transformer-based meta-classifier generator $\cG$, while \(K\) denotes the numbers of neighbors selected by $\cS$. We observe that smaller values for \(K \in \{ 2,3,6\}\) and \(D \in \{1,2\}\) yield superior results. We observe that setting \(D=1\) and \(K=3\) works reasonably well, balancing performance and the computational overhead in learning complex meta-classifier generators.}
\label{tab:combined-ablations}
\begin{tabular}{P{0.07cm}c|ccc}
\hline
\toprule
&\textbf{Ablations}
 & \textbf{P@1} \(\uparrow\) & \textbf{P@5} \(\uparrow\) & \textbf{R@10} \(\uparrow\) \\
\midrule
& \alg ($D=1$, $K=3$) & 69.29 & 38.81 & 80.40 \\
\midrule\midrule
\multirow{3}{*}{\rotatebox{90}{{\small\underline{{Generator} \((\cG)\)}}}} &  $D=2,~K=3$ & 70.36 & 39.06 & 80.39 \\
& $D=4,~K=3$ & 70.71 & 39.11 & 80.27 \\[2pt]
\cline{2-5}
&$\cG$ as Sum, $K=3$ & 45.49 & 25.46 & 59.01   \\
&$\cG$ as wt. Sum, $K=3$ & 46.39 & 25.88 & 59.29  \\
\midrule
\midrule
\multirow{4}{*}{\rotatebox{90}{{\small\underline{{Selector} \((\cS)\)}}~}} &$D=1,~K=1$ & 68.98 & 38.56 & 80.11  \\
&$D=1,~K=2$ & 69.34 & 38.74 & 80.25  \\
&$D=1,~K=6$ & 69.99 & 39.12 & 80.79 \\
&$D=1,~K=20$ & 69.07 & 38.57 & 79.80   \\
\bottomrule
\end{tabular}
\end{table}

\begin{table}[!t]
\centering
\fontsize{7.5}{10}\selectfont
\caption{Performance comparison of NGAME and NGAME+\alg as the percentage of \new items (Novel Ratio) in the WikiHierarchy dataset is varied from 10\% to 40\%. While NGAME exhibits a significant decrease, NGAME+\alg is relatively more resilient.}
\label{tab:vary-ratio}
\begin{tabular}{@{}c|ccc|ccc@{}}
\toprule
\multirow{2}{*}{\textbf{Novel Ratio}}  & \multicolumn{3}{c|}{\textbf{NGAME}} & \multicolumn{3}{c}{\textbf{NGAME+IRENE}} \\ \cline{2-7}
& \textbf{P@1} \(\uparrow\) & \textbf{P@5} \(\uparrow\) & \textbf{R@10} \(\uparrow\) & \textbf{P@1} \(\uparrow\) & \textbf{P@5} \(\uparrow\) & \textbf{R@10} \(\uparrow\) \\ 
\midrule
10\% & 66.19 & 59.25 & 27.08 & 91.33 & 86.67 & 40.09 \\  
20\% & 65.29 & 58.22 & 26.66 & 89.37 & 85.29 & 39.33 \\
30\% & 63.59 & 57.12 & 26.09 & 89.50 & 84.64 & 38.48 \\
40\% & 60.12 & 53.47 & 24.45 & 87.64 & 82.59 & 36.99 \\
\bottomrule
\end{tabular}
\vskip-10pt
\end{table}

\noindent \textbf{Inputs to the Meta-classifier Generator \(\cG\):} Keeping the classifier selector \(\cS\) and meta-classifier generator \(\cG\) architecture fixed, the inputs passed to \(\cG\) were changed to encoder representations $\vz$ of the selected \seen items, instead of their classifiers \(w_{\ell}\). \alg with classifiers \(w_{\ell}\) of the selected \seen items yields superior performance, with a significant margin of 4\% on P@$1$. This underscores the value of leveraging high-capacity classifiers of the selected items, which contain beneficial information for representing a \new item. \par
Detailed discussions regarding these ablations are present in Appendix~\ref{sec:supp_ablations_discussions}.

\noindent \textbf{Ratio of Novel Items:}
As we vary the percentage of \new items in the LF-WikiHierarchy-550K dataset from 10\% to 40\%, We observe that NGAME+\alg is relatively more resilient to changes in the ratio of \new items, in comparison to the base NGAME algorithm. These results are summarized in Table~\ref{tab:vary-ratio}.\par 

\noindent \textbf{Adapting to Few-Shot Scenarios - \alg-OneShot:}
A plethora of algorithms have been proposed to address few-shot scenarios, with a particular emphasis on mitigating catastrophic learning. We study the extreme case of this where only a single ground truth query is revealed for an item. 
\alg leverages the revealed query of the one-shot item to fetch the nearest classifiers alongside the item itself. A max-voting strategy is then employed to improve the selection of observed classifiers over scenarios where only the item text is considered for this purpose. This variant is denoted as \alg-OneShot. Remarkably, \alg-OneShot can exploit the revealed data to improve the prediction accuracy without having the need to fine-tune the base model.
To provide a comprehensive defense against catastrophic forgetting, we explore a theoretical baseline where NGAME is trained on a consolidated dataset comprising the initial training data for observed items and the data made available for the one-shot items. Even when compared to this ideal and optimal baseline, \alg-OneShot was found to be superior in terms of P@1 by about 1\%. Further, compared to Semsup-XC, which finetunes its encoder on the revealed data, \alg-OneShot was at least 10\% better in R@10. These results highlight \alg-OneShot's ability to surpass models that fine-tune their encoders, while preserving the base encoder and reducing latency and complexity (cf. Table~\ref{tab:table_one_shot}).

\begin{table}[!t]
\centering
\fontsize{7}{10}\selectfont
\caption{Performance evaluation of (NGAME+) \alg-OneShot, an extension for one-shot retrieval. One randomly selected query is made available to all algorithms for the \new items, and evaluation is done solely on the one-shot items. NGAME-retrained involves training NGAME method from scratch on the original training corpus along with the revealed queries for previously unseen items. Even with a significant training overhead for the NGAME-retrained model ($\sim$82 hours on LF-AmazonTitles-1.3M-10 and $\sim$43 hours on LF-Wikipedia-500K-10), IRENE consistently outperforms the NGAME-retrained model by 1-2\% across all metrics without incurring any additional training overhead. This is achieved by incorporating the new queries into the classifier selector module (at inference time) to refine the classifier shortlist.}
\label{tab:table_one_shot}
\begin{tabular}{@{}c|ccc|ccc@{}}
\toprule
\multirow{2}{*}{\textbf{Method}}   & \multicolumn{3}{c}{\textbf{LF-AmazonTitles-1.3M-10}} & \multicolumn{3}{c}{\textbf{LF-Wikipedia-500K-10}} \\ \cline{2-7}
& \textbf{P@1}  & \textbf{P@5} & \textbf{R@10}  & \textbf{P@1} & \textbf{P@5} & \textbf{R@10} \\ 
\midrule
NGAME-retrained & 30.55 & 15.39 & 36.48 & 47.37 & 15.33 & 66.37 \\  
SemSup-XC-OneShot & 13.96 & 6.98 & 19.76 & 46.25 & 13.77 & 57.04 \\
\alg-OneShot & \textbf{31.92} & \textbf{16.67} & \textbf{39.72} & \textbf{47.90} & \textbf{15.79} & \textbf{68.69} \\
\bottomrule
\end{tabular}
\vskip-3pt
\end{table}

%% file: tables/main_zero_shot_gen_zero_shot.tex
\begin{table*}[!t]
\centering
\fontsize{7}{8}\selectfont
\caption{Comparison of zero-shot and generalized zero-shot accuracies of \alg when applied to various baseline encoder frameworks. On average, \alg improves P@1 by 10.1\% and R@10 by 11.9\% in the zero-shot setting. In the generalized setting \alg boosts P@1 by 15.5\% and R@10 by 11.5\%. The best-performing algorithm amongst each pair of base encoder and its \alg variant are indicated in boldfaced.}
\label{tab:main-zs-gzs-acc}
\vskip-7pt
\begin{tabular}{c|cccc|cccc|cccc|cccc}
\toprule
\multirow{4}{*}{\textbf{Model}} & \multicolumn{4}{c|}{\textbf{LF-AOL-270K-10}} & \multicolumn{4}{c|}{\textbf{LF-WikiHierarchy-550K-10}} & \multicolumn{4}{c|}{\textbf{LF-AmazonTitles-1.3M-10}} & \multicolumn{4}{c}{\textbf{LF-Wikipedia-500K-10}} \\
\cmidrule(lr){2-5} \cmidrule(lr){6-9} \cmidrule(lr){10-13} \cmidrule{14-17}
 & \multicolumn{2}{c}{\textbf{Zero-Shot}} & \multicolumn{2}{c|}{\textbf{Generalized}} & \multicolumn{2}{c}{\textbf{Zero-Shot}} & \multicolumn{2}{c|}{\textbf{Generalized}} & \multicolumn{2}{c}{\textbf{Zero-Shot}} & \multicolumn{2}{c|}{\textbf{Generalized}} & \multicolumn{2}{c}{\textbf{Zero-Shot}} & \multicolumn{2}{c}{\textbf{Generalized}} \\
\cmidrule(lr){2-3} \cmidrule(lr){4-5} \cmidrule(lr){6-7} \cmidrule(lr){8-9} \cmidrule(lr){10-11} \cmidrule(lr){12-13} \cmidrule(lr){14-15} \cmidrule(lr){16-17}
& \textbf{P@1} & \textbf{R@10} & \textbf{P@1} & \textbf{R@10} & \textbf{P@1} & \textbf{R@10} & \textbf{P@1} & \textbf{R@10} & \textbf{P@1} & \textbf{R@10} & \textbf{P@1} & \textbf{R@10} & \textbf{P@1} & \textbf{R@10} & \textbf{P@1} & \textbf{R@10} \\
\midrule
NGAME & 30.90 & 54.20 & 20.16 & 38.27 & 46.01 & 58.66 & 66.19 & 27.08 & 30.42 & 36.44 & 45.14 & 30.25 & \textbf{46.96} & 65.27 & \textbf{81.86} & \textbf{69.58} \\
NGAME+\alg & \textbf{36.47} & \textbf{59.57} & \textbf{35.11} & \textbf{52.30} & \textbf{69.29} & \textbf{80.40} & \textbf{91.33} & \textbf{40.09} & \textbf{31.56} & \textbf{38.83} & \textbf{47.77} & \textbf{31.49} & 44.91 & \textbf{67.79} & 78.99 & 69.27 \\
\cmidrule{1-17}
ANCE & 33.43 & 67.84 & 22.63 & 49.72 & 43.06 & 56.28 & 68.76 & 25.89 & 22.38 & 30.72 & 27.65 & 17.31 & 30.67 & 58.91 & 42.91 & 43.39 \\
ANCE+\alg & \textbf{36.84} & 67.82 & \textbf{30.84} & \textbf{51.75} & \textbf{66.54} & \textbf{82.10} & \textbf{90.72} & \textbf{39.74} & \textbf{22.75} & \textbf{32.72} & \textbf{36.78} & \textbf{22.34} & \textbf{41.59} & \textbf{71.59} & \textbf{71.39} & \textbf{63.46} \\
\cmidrule{1-17}
MACLR & 11.31 & 18.24 & 9.26 & 7.52 & 30.37 & 35.47 & 59.44 & 14.31 & \textbf{21.93} & 28.59 & 27.50 & 15.99 & 39.56 & 68.53 & 46.59 & 46.62 \\
MACLR+\alg & \textbf{34.32} & \textbf{61.29} & \textbf{30.40} & \textbf{44.71} & \textbf{69.45} & \textbf{81.43} & \textbf{88.81} & \textbf{38.37} & 21.56 & \textbf{28.77} & \textbf{31.49} & \textbf{18.85} & \textbf{44.64} & \textbf{73.05} & \textbf{70.52} & \textbf{62.82} \\
\cmidrule{1-17}
DPR & 30.38 & 53.82 & 19.71 & 37.99 & 44.84 & 59.29 & 65.19 & 26.73 & \textbf{31.10} & \textbf{40.98}  & 38.18 & 26.93 & \textbf{42.90} & \textbf{71.20} & 51.54 & 61.84 \\
DPR+\alg & \textbf{36.80} & \textbf{60.22} & \textbf{35.07} & \textbf{52.57} & \textbf{69.65} & \textbf{80.01} & \textbf{89.52} & \textbf{39.84} & 30.49 & 40.31 & \textbf{43.08} & \textbf{29.25} & 42.19 & 70.50 & \textbf{70.39} & \textbf{66.71} \\
\midrule
\midrule
TF-IDF & 13.74 & 28.05 & 6.61 & 9.90 & 22.79 & 21.44 & 64.50 & 10.88 & 8.12 & 24.15 & 16.33 & 20.40 & 11.53 & 23.89 & 15.07 & 14.79 \\
Zest-XML & 9.34 & 25.91 & 26.34 & 26.57 & 13.97 & 17.48 & 68.86 & 22.13 & 5.42 & 5.58 & 41.36 & 22.87 & 2.62 & 14.73 & 60.16 & 45.15 \\
SemSup-XC & 26.27 & 36.31 & 26.12 & 23.92 & 57.45 & 46.81 & 90.51 & 28.37 & 11.28 & 11.68 & 25.13 & 15.21 & 46.60 & 57.08 & 54.20 & 38.08 \\
DEXA & 21.68 & 41.85 & 25.09 & 46.76 & 54.83 & 66.89 & 76.18 & 36.17 & 28.83 & 35.19 & 48.19 & 30.89 & 42.76 & 67.37 & 67.98 & 65.86 \\
\bottomrule
\end{tabular}
\vskip-5pt
\end{table*}

%% file: sections/8_conclusions.tex
\section{Conclusions}

In this paper, we studied large-scale zero-shot retrieval and developed techniques to efficiently and accurately represent novel items. We proposed \xmec, a generic algorithmic framework for learning accurate meta-classifiers for \new items, wherein different architectural and training choices can yield algorithms with varying accuracy and efficiency trade-offs, presenting ample opportunities for future research. We also developed \alg, a simple and practically deployable instance that significantly boosts the performance of any Siamese encoder with minimal overheads, along with a novel theoretical framework for analyzing generalization in large-scale zero-shot retrieval. Comprehensive empirical validation and online A/B tests in a \textit{Sponsored Search} application on a major search engine demonstrated the utility of \alg and \xmec.

%% file: sections/9_appendix.tex
\onecolumn
\begin{appendix}

\section{Proofs of Theorems} \label{App:Proofs}
We now present the proofs of the various theorems presented in the main manuscript. We recall the definitions presented in Section~\ref{Sec:Theory}. The data can be represented as query-item pairs, given by \(\vs = (\vx,\vz)\).  A key design choice is on forming the dataset \(\sS = \{\vs\}\). The simplest formulation is that of selecting \(N\) queries and items randomly from the underlying distributions \(\cX\) and \(\cZ\) respectively, and pairing them, to form the dataset \(\tilde{\sS} = \{\vs_i = (\vx_i,\vz_i)~|~i=1,2,\ldots N\}\). Alternatively, given \(\sX\) and \(\sZ\), defined as \(\sX = \{\vx_i\}_{i=1}^N\), and \(\sZ_s = \{\vz_{\ell}\}_{\ell=1}^L\) respectively, we can consider all possible pairing of queries and items, giving us \(\sS = \{ (\vx_i,\vz_\ell)~|~i=1,2,\ldots N, \ell = 1, 2,\ldots L\}\), \it{i.e.,} \(\vs\) is drawn from the Cartesian product space \(\gS = \gX\times\gZ\). Note that \(\sS\) can be indexed as \( j = L(i-1) + \ell,~j=1,2,\ldots M = NL\). The extreme (meta) classification problem is now one of binary classification, with targets \(y\) drawn from the space \(\gY\), that satisfy \(y_{j} = 1\) if item \(\vz_{\ell}\) is positively associated with the query \(\vx_i\). In the XC setting, it is well known that negatively associated pairs are significantly more likely to occur than positive ones. We  assume that in any set \(\sS\), we have at most \(\kappa\) positively associated pairs.
Let \(\max_{\vx\in\gE(\gX)} \|\vx\|_2 \leq B \) and \(\max_{\vw\in\sW} \|\vw\|_2 \leq W \) be the bounds on the norm of the encoder representations of the queries and the learnt classifiers, respectively. \par

Before we proceed with the proofs, we recall a few key definitions. \par
\noindent {\bf Empirical and True Risk}: Given a function \(f(\cdot)\) drawn from a function space \(\cF\), and a loss function, the empirical risk over \(\sS\), and the true risk, are given by
\begin{align*}
    \hat{R} = \frac{1}{M} \sum_{\subalign{j=1\\\vs_j\sim\sS}}^{M} \mathrm{loss}(f(\vs_j),y_j)\quad\text{and}\quad R = \mathbb{E}_{\vs\sim\gX\times\gZ}\left[ \mathrm{loss}(f(\vs),y) \right],\quad\text{respectively.}
\end{align*}
\noindent {\bf Rademacher Complexity}~\citep{FML12}: Given a function \(f\) drawn from the function class \(\cF\), the empirical Rademacher complexity defined over the set \(\sS\), and the Rademacher complexity over all sets of size \(M\) are given by:
\begin{align}
    \hat{\mathfrak{R}}_{\sS}(\cF) = \frac{1}{M}\mathbb{E}_{\bm{\sigma}}\left[ \sup_{f\in\cF} \sum_{j=1}^{M} \sigma_j f(\vs_j) \right],\quad\text{and}\quad\mathfrak{R}_M(\cF) = \mathbb{E}_{\sS}\left[\hat{\mathfrak{R}}_{\sS}(\cF)\right],
    \label{Eqn:risk}
\end{align}
respectively, where \(\bm{\sigma} = (\sigma_1, \sigma_2,\ldots,\sigma_M)\), whose entries are independent random variables drawn from the \it{Rademacher distribution} i.e. \(\mathrm{Pr}(\sigma_j = +1) = \mathrm{Pr}(\sigma_j = -1) = 1/2,~\forall j\). \par
\noindent {\bf McDiarmid Inequality}~\citep{combes2023extension}:  McDiarmid's inequality is a concentration inequality to bound the deviation of the sampled value from the expected value. We consider an extension of the standard inequality, to the case where the function under consideration \(\Phi\) does not strictly satisfy the bounded differences property, but large differences remain very rare. \par
Let \(\Phi: \mathcal{S}_1 \times \mathcal{S}_2 \times \cdots \times \mathcal{S}_M  \rightarrow \mathbb{R}\) be a function acting on the dataset \(\sS\), with \(\vs_j \in \mathcal{S}_j\) and \(\mathcal{S}^\prime \subseteq \mathcal{S}_1 \times \mathcal{S}_2 \times \cdots \times \mathcal{S}_M\) be a subset of its domain and let \(c_1, c_2, \dots, c_n \ge 0\) be constants such that all pairs \((\vs_1,\ldots,\vs_M)\in \mathcal{S}^\prime\) and \((\vs'_1,\ldots,\vs'_n)\in \mathcal{S}^\prime\), satisfy the bounded difference property:
\[
\left|\Phi(\vs_1, \ldots, \vs_M) - \Phi(\vs'_1, \ldots, \vs'_M)\right| \leq \sum_{j: \vs_j \ne \vs'_j} c_j.
\]
All datasets drawn from \(\mathcal{S}^\prime\) can be viewed as \it{good sets} that satisfy bounded difference. Then, given a random dataset \(\sS\), with probability \(q = 1 - \mathrm{Pr}(\sS \in \mathcal{S}^\prime)\) (which is the probability of drawing a \it{bad set}), the following holds:
\[
\mathrm{Pr}\left( \left|\Phi(\sS) - \mathbb{E}_{\sS}\left[\Phi(\sS)\right] \right|\geq \epsilon \right) \leq 2q + 2\exp \left\{ - \frac{2\left(\max\left\{0,\epsilon - q \sum_j c_j\right\}\right)^2}{\sum_j c_j^2}\right\}.
\]
\noindent {\bf Hoeffding's Inequality}~\citep{FML12}: The Hoeffding's inequality is a special case of the McDiarmid inequality. Let \(Y_1, Y_2, \ldots Y_M\) be \(M\) independent Bernoulli-distributed random variables. Let the sum be denoted as \(S_M = Y_1 + Y_2 + \ldots + Y_M\). Then, \(\forall t>0\), we have
\[
\mathrm{Pr}\left(S_M - \mathbb{E}_Y\left[S_M\right] \geq t \right) \leq \exp\left\{ - \frac{2t^2}{M}\right\}.
\]
\noindent {\bf Proof of Theorem~\ref{Theorem:GenBound}}: Consider the empirical and true risk defined in Equation~\ref{Eqn:risk}. The proof follows by deriving the generalization bound using the McDiarmid inequality, followed by using the Hoeffding inequality to bound the probability \(q\). Without loss of generality, we redefine the loss to have the target labels \(y_{j} \in \{-1,1\}\), giving rise to the following form of the loss:
\begin{align}
    g(\vs,y) = \mathrm{loss}\left(f(\vs),y\right) = \left(\frac{1-y}{2}\right)f(\vs) - C \left(\frac{1+y}{2}\right)f(\vs).
    \label{Eqn:loss}
\end{align}
Let \(\Phi(\sS)\) be defined as follows
\begin{align*}
    \Phi(\sS)= \sup_{f\in\cF} \left\{ \bbE_{\vs}\left[ g(\vs,y)\right]- \frac{1}{M} \sum_{j=1}^{M} g(\vs_j,y_j) \right\} = \sup_{f\in\cF} \left\{ \bbE\left[ g(\vs,y)\right]- \hat{\bbE}\left[ g(\vs_j,y_j)\right] \right\}
\end{align*}
The subsequent steps are consistent with those in deriving the Rademacher complexity in the standard setting~\citep{FML12}, but modified to accommodate the extension of the McDiarmid inequality. Let \(\sS\) and \(\sS'\) be two sample sets, differing only in element \(\vs_k\). Then, we have:
\begin{align*}
    \Phi(\sS) - \mathbb{E}_{\sS}\left[\Phi(\sS)\right] \leq \frac{1}{M} \sup_{f\in\cF} \left\{ g(\vs_k,y_k) - g'(\vs_k,y_k)\right\} \leq c_j
\end{align*}
if \(\cS\) and \(\cS'\) are drawn from the \it{good set} \(\tilde{\mathcal{S}}\) (which we discuss shortly), then, we have \(c_j = \frac{1}{M},~\forall j\), and McDiarmid's inequality can be applied to \(\Phi(\sS)\) to obtain the following, for \(\epsilon > q\):
\begin{align*}
\mathrm{Pr}\left( \left|\Phi(\sS) - \mathbb{E}_{\sS}\left[\Phi(\sS)\right]\right| \geq \epsilon \right) &\leq 2q + 2\exp \left\{ - \frac{2\left(\max\left\{0,\epsilon - q \sum_j c_j\right\}\right)^2}{\sum_j c_j^2}\right\}\\
&=2q + 2\exp \left\{ - 2M\left(\max\left\{0,\epsilon - q\right\}\right)^2\right\} \\
&= \underbrace{2q + 2\exp \left\{ - 2M\left(\epsilon - q\right)^2\right\}}_{\delta}
\end{align*}
Solving for \(\epsilon\) on the right hand side, we obtain:
\begin{align*}
    \delta &= 2q +2\exp \left\{ - 2M\left(\epsilon - q\right)^2\right\}\\
    \Rightarrow 0 &= \epsilon^2 + (-2q)\epsilon + \left(q^2 + \frac{1}{2M}\ln\left(\frac{2}{\delta-2q}\right) \right) \\
    \Rightarrow \epsilon &= q + \sqrt{\frac{\ln\left(\frac{2}{\delta-2q}\right)}{2M}},
\end{align*}
for \(\delta \in (2q,1)\). Substituting the above into McDiarmid's inequality, we obtain,  with probability \(1-\delta\),
\begin{align*}
    \Phi(\sS) \leq \mathbb{E}_{\sS}\left[\Phi(\sS)\right] + \left(q + \sqrt{\frac{\ln\left(\frac{2}{\delta-2q}\right)}{2M}}\right).
\end{align*}
It is straightforward to derive the bound \(\mathbb{E}_{\sS}\left[\Phi(\sS)\right] \leq 2 \mathfrak{R}_M(\mathrm{loss} \circ\cF)\)~\citep{FML12}. Similarly, by applying the McDiarmid inequality to \(\mathfrak{R}\), we get:
\begin{align*}
    \mathfrak{R}_M(\mathrm{loss} \circ\cF) \leq \hat{\mathfrak{R}}_{\sS}(\mathrm{loss} \circ\cF) + \left(q + \sqrt{\frac{\ln\left(\frac{2}{\delta-2q}\right)}{2M}}\right).
\end{align*}
Further, for the given form of the loss, we have the following result:
\begin{align*}
    \hat{\mathfrak{R}}_{\sS}(\mathrm{loss} \circ\cF) &= \frac{1}{M}\mathbb{E}_{\bm{\sigma}}\left[ \sup_{f\in\cF} \sum_{j=1}^{M} \sigma_j (\mathrm{loss}\circ f)(\vs_j) \right] \\
    &= \frac{1}{M}\mathbb{E}_{\bm{\sigma}}\left[ \sup_{f\in\cF} \sum_{j=1}^{M} \sigma_j \left( \left(\frac{1-y_j}{2}\right)f(\vs_j) - C \left(\frac{1+y_j}{2}\right)f(\vs_j)\right) \right] \\
    &= \frac{1}{2M}\mathbb{E}_{\bm{\sigma}}\left[ \sup_{f\in\cF} \sum_{j=1}^{M}  \sigma_jf(\vs_j)(1-C) - \sigma_jy_jf(\vs_j)(1+C) \right]
\end{align*}
Since \(\sigma_j\) and \(-\sigma_jy_j\) are distributed identically, we have: 
\begin{align*}
    \hat{\mathfrak{R}}_{\sS}(\mathrm{loss} \circ\cF) &= \frac{1}{2}\mathbb{E}_{\bm{\sigma}}\left[ 2 \frac{1}{M}\sup_{f\in\cF} \sum_{j=1}^{M}  \sigma_jf(\vs_j)\right] = \hat{\mathfrak{R}}_{\sS}(\cF)
\end{align*}
Putting the above together, and substituting in for \(\Phi(\sS)\), we obtain
\begin{align*}
    R &\leq \hat{R} + \hat{\mathfrak{R}}_{\sS}(\cF) + 3\left(q + \sqrt{\frac{\ln\left(\frac{2}{\delta-2q}\right)}{2M}}\right),
\end{align*}
which completes part of the proof of Theorem~\ref{Theorem:GenBound}. \par 
What remains, is to derive \(q\), the probability of drawing a \it{bad set} \(\sS\). A \it{good set} is one which satisfies the bounded difference condition. Consider the set \(\sS = \{\vs_1, \vs_2, \ldots, \vs_M\}\). Let \(p\) denote the probability that a query-item pair is positively correlated. Then, by the definition of the loss in Equation~\ref{Eqn:loss}, for positively correlated pairs, the maximum error (over all \(f \in \cF\)) is W, while for negatively correlated pairs, it is 1. For \(c_j = \frac{1}{M}\), we require that the sample \(\vs_j\) is a negative pair. 
Let a \it{good set} \(\sS\) have at most \(\kappa\) positively associated pairs. Then, a \it{bad set} contains at least \(M-\kappa\) positively associated pairs. Then, \(q\) is the probability of drawing a set \(\sS\) with at least \(M-\kappa\) positively associated pairs. \par 

To derive this probability, consider indicator variables \(Y_j = \mathbb{I}_{[\vs_j\text{ contains a positive pairing}]}\). Let \(S_M = \sum_j Y_j\) denote the sum. Then, we have \(\bbE_Y\left[S_M\right] = pM\). Applying the Hoeffding inequality with \(t = M - \kappa - pM\), we get:
\begin{align*}
    \mathrm{Pr}\left( S_M - \bbE_Y\left[S_M\right] \geq t \right) & \leq 2\exp\left\{ -\frac{2t^2}{M}\right\} \\
    \Rightarrow 
    \mathrm{Pr}\left( S_M - pM \geq M-\kappa-pM \right) = \mathrm{Pr}\left( S_M  \geq M-\kappa \right) & \leq 2\exp\left\{ -\frac{2(M-\kappa-pM)^2}{M}\right\}.
\end{align*}
Simplifying, we get:
\begin{align*}
    q = \mathrm{Pr}\left(\text{set \(\sS\) contains at least \(M-\kappa\) positively associated pairs}\right) \leq  2\exp\left\{ -2M\left(1-p-\frac{\kappa}{M}\right)^2\right\},
\end{align*}
which is significantly smaller than 1, given large \(M\) (which is true in the case of XC, where we have large \(N\) and \(L\)) and \(\kappa = 0\), which is the setting under which the required McDiarmid inequality is defined. This completes the proof of Theorem~\ref{Theorem:GenBound}.\\\par

\noindent {\bf Proof of Lemma~\ref{Lemma:XC}}: We now state and discuss the proof of Lemma~\ref{Lemma:XC}, which is an extension of Theorem~3 from~\citet{awasthi20adv}.
\begin{lemma*}
\label{App_Lemma:XC}
({\bf Rademacher complexity of the XC Classifiers}) (extension of~\citet{awasthi20adv}, Theorem~3) Let \(\cF\) be the class of linear classifiers defined over the seen-item set \(\sZ_s\) in the classical XC setting (cf. Section~\ref{SubSec:Prelims}), \it{i.e.,} \(\cF = \{\langle \vx,\vw_{\ell}\rangle~|~\ell=1,2,\ldots,L\}\), where \(\vx \in \sX\). Then, the Rademacher complexity of \(\cF\) can be bounded as follows:
\begin{align*}
    \hat{\mathfrak{R}}_{\sS}(\cF) \leq \frac{LBW}{\sqrt{N}},
\end{align*}
where \(|\sX| = N\) is the Cardinality of the training set.
\end{lemma*}
\begin{proof}
    Let \(\cF_{\ell} = \{ \langle \vx,\vw_{\ell} | \||\vw_{\ell}\|_2 \leq W\}\).~\citet{awasthi20adv} showed that \(\displaystyle \hat{\mathfrak{R}}_{\sS}(\cF_\ell) \leq \frac{W}{N}\|\vX\|_F\), where \(\vX\) is a matrix formed with the elements of \(\sX\). For bounded data, \(\max_{\vx\in\gE(\gX)} \|\vx\|_2 \leq B \), we have \(\displaystyle \hat{\mathfrak{R}}_{\sS}(\cF_\ell) \leq \frac{WB}{\sqrt{N}}\). Given \(\cF\) with \(L\) classifiers, the Rademacher complexity is bounded by \(\displaystyle \hat{\mathfrak{R}}_{\sS}(\cF_\ell) \leq \frac{LBW}{\sqrt{N}}\).
\end{proof} 
\noindent {\bf Proof of Lemma~\ref{Lemma:IRENE} and Corollary~\ref{Corollary:IRENE}}: We now derive the bound on the Rademacher complexity of the \alg meta-classifier generator. Recall the Lemma
\begin{lemma*}
\label{App_Lemma:IRENE}
({\bf Rademacher complexity of the \alg generator}) Let \(\cF\) be the class of functions defined in the \alg algorithm, comprising pre-determined encoder representations and classifiers, a given classifier selector that outputs \(K\) classifiers, and \(\cG\), the meta-classifier generator.  Then, the Rademacher complexity of \(\cF\) can be bounded as follows:
\begin{align*}
    \hat{\mathfrak{R}}_{\sS}(\cF) \leq \bigO{B \|\vM\|_2 \sqrt{d \ln(K+1)}},
\end{align*}
where \(\vx\in\sR^d\) and \(\vM \in \sR^{d\times1}\) is the weight matrix associated with the linear layer.
\end{lemma*}
\begin{proof}
    Recall the \alg meta-classifier generator:
\begin{align*}
    \cF &= \left\{(\vx_i\,,\,\vz_\ell) \rightarrow \left\langle \vx_i , \mathrm{Linear}\left(\mathrm{Self Attention}\left( \cC_{lf} \left( \cS(\vz_\ell)\right) \right) \right) \right\rangle \right\},\quad\text{where} \\
    \cF_3 &= \cC_{lf} \left( \cS(\vz_\ell)\right),~\cF_2 = \mathrm{Self Attention}(\cF_3),~\text{and}~\cF_1 = \mathrm{Linear}(\cF_2) = \vM \cF_2 + \vb
\end{align*}
The proof follows by repeatedly applying Talagrand’s lemma, which states that, given an \(L\)-Lipschitz continuous function \(g\), and the function class \(\cF\), we have \(\hat{\mathfrak{R}}_{\sS}(g \circ \cF) \leq L\, \hat{\mathfrak{R}}_{\sS}(\cF)\). Applying Talagrand’s lemma to \(\cF\), we get:
\begin{align}
    \hat{\mathfrak{R}}_{\cS}(\cF) &\leq B~ \hat{\mathfrak{R}}_{\cS}(\cF_1) \label{Eqn:Lip1}\\
    &\leq B \|\vM\|_2~\hat{\mathfrak{R}}_{\cS}(\cF_2) \label{Eqn:Lip2}\\
    &\leq B \|\vM\|_2 \sqrt{d\ln\left(K+1\right)}~\hat{\mathfrak{R}}_{\cS}(\cF_3) \label{Eqn:Lip3}\\
    &\leq \bigO{B \|\vM\|_2 \sqrt{d\ln\left(K+1\right)}}, \label{Eqn:Lip4}
\end{align}
where Equation~\ref{Eqn:Lip1} is a consequence of the boundedness of \(\vx\), Equation~\ref{Eqn:Lip2} is obtained by considering the Lipschitz constant of a linear transformation layer with vector-valued weights, and Equation~\ref{Eqn:Lip3} is obtained by applying the Lipschitz constant of a self-attention layer, derived by~\citet{vuckovic2020SelfAttnLip}, where in turn, \(d\) is the dimensionality of the input sequence, and \(K\) is the number of classifiers selected by \(\cS\). We note that, while~\citet{Kim21SelfAttnLip} provide a tighter bound in the context of \(L_2\) multi-head attention, their result does not hold for the dot-product-based self-attention block considered above. Equation~\ref{Eqn:Lip4} is a consequence of considering a fixed set of classifiers \(\sW\), and a pre-determined classifier selector algorithm. Therefore, these blocks merely add a constant factor to the complexity to the model \(\cF\) class, completing the proof of Lemma~\ref{Lemma:IRENE}.
\end{proof}
The analysis can be extended to derive Corollary~\ref{Corollary:IRENE} by incorporating trainable classifiers. Recall the statement of the Corollary:
\begin{corollary*}
({\bf Rademacher complexity of the \alg generator with trainable classifier}) Let \(\cF\) be the class of functions defined in the \alg algorithm as in Lemma~\ref{Lemma:IRENE}. Let the classifier set \(\sW\) be trainable over the meta-classifier loss. Then, the Rademacher complexity of \(\cF\) can be bounded as follows:
\begin{align*}
    \hat{\mathfrak{R}}_{\sS}(\cF) \leq \bigO{B^2W \sqrt{\frac{L}{M}} \|\vM\|_2 \sqrt{d \ln(K+1)}},
\end{align*}
where \(\vx\in\sR^d\) and \(\vM \in \sR^{d\times1}\) is the weight matrix associated with the linear layer.
\end{corollary*}
\begin{proof}
    The above result can be obtained by extending the proof of Lemma~\ref{Lemma:IRENE}:
\begin{align*}
\hat{\mathfrak{R}}_{\cS}(\cF) &\leq B \|\vM\|_2 \sqrt{d\ln\left(K+1\right)}~\hat{\mathfrak{R}}_{\cS}(\cF_3) \nonumber\\
&\leq \bigO{B \|\vM\|_2 \sqrt{d\ln\left(K+1\right)}~\hat{\mathfrak{R}}_{\cS}(\cF_{\mathrm{clf}})},
\end{align*}
where \(\hat{\mathfrak{R}}_{\cS}(\cF_{\mathrm{clf}})\) can be obtained from Lemma~\ref{Lemma:XC}, considering the dataset \(\sS\) with \(M=NL\) samples, and \(L\) classifiers associated with the observed items, which yields the desired result:
\begin{align*}
\hat{\mathfrak{R}}_{\cS}(\cF) &\leq \bigO{B \|\vM\|_2 \sqrt{d\ln\left(K+1\right)} \frac{LW}{M}\|\vX\|_F}\\
&= \bigO{\sqrt{\frac{L}{N}}B^2W \|\vM\|_2 \sqrt{d\ln\left(K+1\right)}}.
\end{align*}
This completes the proof of Corollary~\ref{Corollary:IRENE}.
\end{proof}

\section{Dataset Creation and Statistics}
\label{sec:supp_dataset_creation}
\alg was evaluated on a diverse set of datasets spanning multiple applications like product recommendation (LF-AmazonTitles-1.3M), category annotation (LF-Wikipedia-500K), query completion (LF-AOL-270K), and taxonomy completion (LF-WikiHierarchy-600K). Additionally, \alg was evaluated on a proprietary query-to-keyword matching dataset. Table \ref{tab:table_dataset_statistics} describes the statistics of the different datasets.

Zero-shot splits of the dataset were created in the following manner: Given the original dataset consisting of $L$ items and a split fraction $s$, $sL$ items were randomly selected to form the \unseen item set. The remaining $(1-s)L$ items along with their associated queries form the training corpus. The \unseen test set was created based on connections between test queries and \unseen items, while for the generalized evaluation, we refer to the data source from~\citep{buvanesh2024enhancing} or~\citep{Bhatia16}. We use the notation \textit{dataset}-\textit{split} (e.g., LF-AmazonTitles-1.3M-10 for a 10\% unseen ratio split) to refer to the zero-shot version of the dataset.

\input{tables/dataset_statistics}

\section{Ablations Discussions}
\label{sec:supp_ablations_discussions}
\noindent \textbf{Meta-classifier Generator $\cG$:}  In the context of generating a new item representation $\mathbf{u}_\ell$, let $\mathcal{S}(\mathbf{z}_\ell)$ denote the set of shortlisted seen items for a given item $l$, with $\mathbf{w}_{\ell} \in \mathcal{S}(\mathbf{z}_\ell)$. Here, $\mathbf{w}_{\ell}^{i}$ represents the classifier of the $i^{\text{th}}$ neighbor, and $\mathbf{z}_\ell$ is the encoder representation of the new item.

\textbf{Summation Formulation:}
The summation formulation for $\mathbf{u}_\ell$ is given by aggregating the encoder representation of the new item with the classifiers of its $K$ neighbors as follows:
\begin{equation}
\mathbf{u}_\ell = \left(\mathbf{z}_\ell + \mathbf{w}_{\ell}^{1} + \mathbf{w}_{\ell}^{2} + \ldots + \mathbf{w}_{\ell}^{K}\right)
\end{equation}

\textbf{Weighted Summation Formulation:}
In the weighted summation formulation, $\mathbf{u}_\ell$ is computed by a weighted sum of the encoder representation and the classifiers of its neighbors, where the weights $c^{i}$ are learned parameters:
\begin{equation}
\mathbf{u}_\ell = \left(c^{0} \mathbf{z}_\ell + c^{1} \mathbf{w}_{\ell}^{1} + c^{2} \mathbf{w}_{\ell}^{2} + \ldots + c^{K}\mathbf{w}_{\ell}^{K}\right)
\end{equation}

\noindent \textbf{Classifier Selector $\cS$:} We discuss detailed ablations on the classifier selector. First, we evaluate the effect of changing $K$, the number of \seen items retrieved by $\cS$, given an ANNS-based \(\cS\), on zero-shot performance. From the results presented in Table~\ref{tab:combined-ablations}, we observe that increasing $K$ from 3 to 6 results in a performance decline of approximately 3\% in P@$1$. Further increasing $K$ to 20 continues to decrease performance. This is consistent with observations made in Section~\ref{Sec:Theory}, wherein smaller \(K\) yield a tighter generalization bound, as derived in Lemma~\ref{Lemma:IRENE}. \par
Second, to demonstrate the flexibility of the \alg framework in diverse applications wherein latency is not critical,as a proof of concept, we consider a classifier selector comprising a GPT-4-based re-ranking model that re-ranks the neighbour shortlist obtained from encoder embeddings. This ablation is carried out on a subset of \(10^4\) novel items. We observe P@1 and R@10 improvements by ~1\%, while the representation time increased by \(\mathcal{O}(10^4)\) (cf. Table~\ref{tab:GPTReprTimes}). Given an item, the representation time is the time for re-ranking neighbours and subsequently, generating representations using \(\cG\).

\begin{table*}[!t]
\centering
\fontsize{8}{12}\selectfont
\caption{Ablation study on meta-classifier generator $\cG$ and classifier selector $\cS$ component in NGAME + \alg on LF-WikiHierarchy-550K-10,  LF-AOL-270K-10, LF-AmazonTitles-1.3M-10, and LF-Wikipedia-500K-10 datasets for zero-shot evaluation. The depth \(\mathrm{D}\) denotes the number of layers in the transformer-based meta-classifier generator $\cG$, while \(K\) denotes the numbers of neighbors selected by $\cS$. We observe that smaller values for \(K \in \{ 2,3\}\) and \(D \in \{1,2\}\) yield superior results. Setting \(D=1\) and \(K=3\) works reasonably well for diverse datasets and base encoders.}
\label{tab:combined-ablations-full}
\begin{tabular}{p{0.01cm}P{2.5cm}|P{0.65cm}P{0.65cm}P{0.65cm}|P{1.0cm}P{1.0cm}P{1.0cm}|P{1.0cm}P{1.0cm}P{1.0cm}|P{0.75cm}P{0.75cm}P{0.75cm}}
\toprule
&\multirow{2}{*}{\textbf{Ablation}} & \multicolumn{3}{c|}{\textbf{LF-AOL-270K-10}} & \multicolumn{3}{c|}{\textbf{LF-WikiHierarchy-550K-10}} & \multicolumn{3}{c|}{\textbf{LF-AmazonTitles-1.3M-10}} & \multicolumn{3}{c}{\textbf{LF-Wikipedia-500K-10}} \\
\cline{3-14}
&&\textbf{P@1} & \textbf{P@5} & \textbf{R@10} & \textbf{P@1} & \textbf{P@5} & \textbf{R@10} & \textbf{P@1} & \textbf{P@5} & \textbf{R@10} & \textbf{P@1} & \textbf{P@5} & \textbf{R@10} \\
\midrule
&\alg ($D=1$, $K=3$) & 36.47 & 11.46 & 59.57 & 69.29 & 38.81 & 80.40 & 31.56 & 16.57 & 38.83 & 44.91 & 15.60 & 67.79 \\
\midrule\midrule
\multirow{2}{*}{\rotatebox{90}{{\small\underline{\(\cG\)}}}} 
& $D=2,~K=3$ & 35.05 & 11.11 & 57.88 & 70.36 & 39.06 & 80.39 & 31.65 & 16.55 & 38.83 & 48.46 & 16.20 & 69.49  \\
& $D=4,~K=3$ & 36.56 & 11.48 & 59.51 & 70.71 & 39.11 & 80.27 & 31.42 & 16.44 & 38.61 & 48.47 & 16.10 & 69.14 \\
\midrule
\midrule
\multirow{4}{*}{\rotatebox{90}{{\small\underline{Selector \((\cS)\)}}}} 
&$D=1,~K=1$ & 36.02 & 11.28 & 58.53 & 68.98 & 38.56 & 80.11 & 31.24 & 16.38 & 38.45 & 45.87 & 15.83 & 68.59 \\
&$D=1,~K=2$ & 36.54 & 11.42 & 59.35 & 69.34 & 38.74 & 80.25 & 31.47 & 16.52 & 38.75 & 45.45 & 15.74 & 68.19 \\
&$D=1,~K=6$ & 35.66 & 11.30 & 58.95 & 69.99 & 39.12 & 80.79 & 31.68 & 16.65 & 39.17 & 45.24 & 15.44 & 67.23 \\
&$D=1,~K=20$ & 36.15 & 11.36 & 59.05 & 69.07 & 38.57 & 79.80 & 31.50 & 16.59 & 39.14 & 45.75 & 15.30 & 66.89 \\
\bottomrule
\end{tabular}
\end{table*}

\begin{table}[h]
    \centering
    \fontsize{9}{9}\selectfont
    \caption{Proof-of-concept experiments on the zero-shot performance comparison of the ANNS-based and GPT-4-based classifier selectors in \alg. When experimented on a small, but random, subset of the data, employing the GPT-4-based re-ranking model results in ~1\% improvements in terms of precision of recall but the time for re-ranking neighbours and generating the representations (Rep. time) increases by an order of four.}
    \label{tab:GPTReprTimes}
    \begin{tabular}{@{}c|ccc@{}}
        \toprule
        Classifier Selector \((\cS)\) & P@1 \(\uparrow\) & R@10 \(\uparrow\) & Rep. Time (ms) \(\downarrow\) \\
        \midrule
        ANNS-based $\cS$    & 67.14 & 85.77 & 0.43   \\
        GPT-4-based $\cS$   & 68.26 & 86.61 & 4000  \\
        \bottomrule
    \end{tabular}
    
\end{table}

\section{Detailed Discussion on Sponsored Search}
\label{sec:supp_sponsored_search}
Sponsored search is essentially a match-making system between users and advertisers with the objective to optimize user experience while searching for knowledge and help advertisers reach the right set of users who might be interested in their product/service~\citep{Aggarwal06}. Users encode their intent in short pieces of text called queries. Similarly, advertisers also bid on short pieces of text, relevant to their ads. One critical component of the sponsored search pipeline is the task of matching user queries to these advertiser bid keywords. Most search engines currently follow different semantics to do this matching called match-types\citep{Broder08}. Matching user queries to advertiser keywords is a nuanced and challenging problem as maintaining the semantics of the match-type is essential to advertisers who often bid differently on different match-types for the same keyword\citep{Rusmevichientong06}. Furthermore, given the placement of this matching task at the forefront of the ads retrieval pipeline, any improvements in accuracy within this application can yield super-linear benefits for downstream components~\citep{Wang15}. Please note that the consideration of a larger truncation factor (30, 50, and 100) for the Precision metric is deliberate. Retrieval algorithms like \alg precede reranking algorithms, necessitating the prediction of a larger candidate set for input into these reranking algorithms.

\textbf{Online Results} 
\alg underwent deployment on a prominent search engine for conducting A/B tests with live search engine traffic. Throughout the live A/B test interaction on the search engine, \alg was systematically compared against an extensive control ensemble featuring diverse algorithms, encompassing not only DR algorithms but also prominent generative, graph-based, and IR algorithms. Evaluation of performance was based on live metrics.   
The findings revealed that \alg led to a 4.2\% increase in the click-through rate (ad clicks obtained per unit query) and a 0.9\% decrease in the quick-back rate (fraction of users quickly leaving the ad landing page due to perceived irrelevance). These results underscore the value creation for users, indicating that \alg effectively presented more relevant ads to the audience. Furthermore, \alg demonstrated a noteworthy 7.8\% increase in keyword density (average number of keywords surviving quality control and relevance filters), affirming the quality of its predictions. 
Additionally, \alg achieved a click efficiency of 150\%, signifying that for every 2\% increase in ad impressions, ads selected by \alg garnered 3\% more clicks. In labeling by expert judges, \alg was found to increase the percentage of good keyword predictions by 9\% (refer to Table~\ref{tab:table_labeling} for details). Notably, \alg successfully matched queries such as "grainger" and "bitwarden" to advertiser keywords like "industrial supply" and "password manager," respectively. It is essential to highlight that these predictions, not relying on text matching, were not replicated by any in-production algorithm. Please refer to Table \ref{tab:flight_examples} for more such predictions made by \alg but missed by the control ensemble.

Furthermore, we conduct a direct comparison of \alg with prominent proprietary and public Dense Retrieval (DR) algorithms currently in production. Specifically, we randomly sample 100 million advertiser keywords introduced into the system after the period covered by the training data scraping. Additionally, we select some of the top-performing dense retrieval encoders deployed in production and pit \alg against them in recommending keywords from this 100 million set for a sampled array of queries. For intellectual property reasons, the names of these algorithms are anonymized, and the results are detailed in Table \ref{tab:table_100M}.
\alg was found to be at least 4\% superior than the next best dense retriever in R@100. As novel items stream into the system, it is be necessary to frequently encode the items and include them in the ANNS index. Table~\ref{tab:inference_breakdown} shows that \alg adds only minimal overhead on top of a language encoder and can get the item representation in less than one ms. Further, the integration of updatable ANNS algorithms, such as Fresh-DiskANN~\citep{singh2021freshdiskann}, can greatly reduce deployment time for novel items.

\begin{table}[!t]
\centering
\fontsize{8}{12}\selectfont
\caption{Results on 100M zero-shot keywords on the search engine. Method Ms are anonymized in-production dense retrievers. All algorithms are provided with just the text of a keyword to get its representation}
\label{tab:table_100M}
\begin{tabular}{l|ccc|ccc}
\toprule
\textbf{Method} & \textbf{R@30} & \textbf{R@50} & \textbf{R@100} & \textbf{P@30} & \textbf{P@50} & \textbf{P@100} \\
\midrule  

M1 & 24.74 & 32.82 & 48.46 & 39.10 & 36.60 & 33.25 \\
M2 & 19.93 & 26.49 & 38.99 & 35.72 & 33.11 & 29.63 \\
M3 & 19.11 & 25.84 & 39.33 & 31.48 & 28.93 & 25.69 \\
M4 & 28.79 & 39.64 & 62.19 & 44.24 & 41.96 & 38.72 \\
NGAME+\alg & \textbf{30.53} & \textbf{42.00} & \textbf{66.68} & \textbf{45.64} & \textbf{43.40} & \textbf{40.11} \\
\bottomrule
\end{tabular}
\end{table}

\begin{table}[!t]
\centering
\fontsize{8}{12}\selectfont
\caption{Expert judges labeling results on KeywordPrediction-10M dataset}
\label{tab:table_labeling}
\begin{tabular}{l|c}
\toprule
\textbf{Method} & \textbf{\% of good quality predictions} \\
\midrule 
NGAME & 64.06 \\
NGAME+\alg & 73.16 \\
\midrule 
NGAME+\alg - OneShot & 77.05 \\
\bottomrule
\end{tabular}
\end{table}

\begin{table}[!t]
\centering
\fontsize{9}{10}\selectfont
\caption{Advertiser keywords predicted by \alg for a user query, but overlooked by the production ensemble comprising leading dense retrieval, graph-based, XC, and generative language models. \alg extends its capability beyond textual similarity by endowing the query representation with world knowledge obtained from classifiers of similar observed items}
\label{tab:flight_examples}
\begin{tabular}{c|c}
\toprule
\textbf{User Query} & \textbf{Advertiser Keyword} \\
\midrule 
best pfmea control plan software & best cmms \\ 
work order app hydraulics & plumbing work order software \\ 
firewalla gold att fiber connection & att ethernet network \\ 
nn2 best home fumigation companies ventura oxnard & bug removal in oxnard \\ 
youtube & streaming services \\ 
netbenefits com login & fidelity retirement account \\ 
financial services crm software & sap customer management software \\ 
hypokalemia & high potassium in the blood \\ 
kaiser options & healthinsurance \\ 
23andme & genetic screening \\ 
\bottomrule
\end{tabular}
\end{table}

\textbf{Offline Results}
To conduct offline experiments, we curated the KeywordPrediction-10M dataset by mining the logs of a commercial search engine within a specific timeframe. The dataset comprised user-typed queries and the corresponding bid keywords for surfaced advertisements, forming query-keyword training pairs. These pairs underwent basic sanity filters based on click-through rate (CTR), clicks, and impressions to generate the training dataset.
Named KeywordPrediction-10M, the dataset encompassed approximately 5 million items and 220 million training queries. For additional details, refer to table \ref{tab:table_dataset_statistics}. As presented in Table \ref{tab:table_5M}, \alg demonstrated a superiority of at least 3\% in Recall@100 and at least 2\% in Precision@30 compared to leading Dense Retrieval (DR) algorithms NGAME and ANCE. 

\begin{table}[!t]
\centering
\fontsize{8}{12}\selectfont
\caption{Results on KeywordPrediction-10M dataset.}
\label{tab:table_5M}
    
\begin{tabular}{l|ccc|ccc}
\toprule
\textbf{Method} & \textbf{R@30} & \textbf{R@50} & \textbf{R@100} & \textbf{P@30} & \textbf{P@50} & \textbf{P@100} \\
& & & & & & \\
\midrule 
 \multicolumn{7}{c}{\textbf{Evaluation only on \new items}} \\\hline\hline
ANCE & 31.37 & 42.87 & 72.45 & 76.3 & 68.40 & 56.32 \\
NGAME & 33.76 & 48.07 & 74.13 &  76.52 & 69.32 & 57.67 \\
NGAME+\alg & \textbf{34.79} & \textbf{49.81} & \textbf{77.32} & \textbf{78.70} & \textbf{71.76} & \textbf{60.18} \\
\hline
Semsup-XC & 34.43 & 49.06 & 74.99 & 77.53 & 70.15 & 58.04 \\
NGAME+\alg-OneShot & \textbf{35.87} & \textbf{50.95} & \textbf{78.46} & \textbf{79.75} & \textbf{72.84} & \textbf{61.20} \\
\hline
\end{tabular}
\end{table}

\textbf{One-shot extension}
We further study the extension of \alg to scenarios involving items that have received precisely one click, representing an exploration of \alg's performance at the extreme tail of observed items. In this context, \alg leverages the revealed query for a one-shot tail item to enhance its classifier selector. The revealed query of the one-shot item is utilized to retrieve the nearest classifiers, along with the item itself. A max-voting strategy is then employed to select superior observed classifiers compared to cases where only the item text is used for this purpose.
In comparison to SemsupXC, which refines its language encoder with new click data, \alg demonstrated a superiority of approximately 3\% and 2\% in Recall@100 and Precision@30, respectively. It's noteworthy that fine-tuning models deployed in production, as undertaken by SemsupXC, introduces latency and complexity costs, making it less desirable. Additionally, algorithms that fine-tune the trained model on revealed data, such as SemsupXC, necessitate rebuilding the Approximate Nearest Neighbors (ANNS) index from scratch. 
In contrast, \alg adopts updatable ANNS algorithms, akin to many Dense Retrieval (DR) algorithms, allowing it to leverage revealed data for improved prediction accuracy without the need to fine-tune the base model. Hence, \alg can make use of the revealed data to improve the prediction accuracy without having to fine-tune the base model. This helps to reduce latency and complexity in an online serving infrastructure.

\input{tables/appendix_full_zero_shot_results}

\input{tables/appendix_full_generalized_zero_shot_results}

\input{tables/appendix_full_ablations_numbers}

\end{appendix}

%% file: tables/dataset_statistics.tex
\begin{table}[h]
\centering
\fontsize{9}{12}\selectfont
\caption{Statistics of different datasets used to benchmark \alg}
\label{tab:table_dataset_statistics}
\vspace{-7pt}
\begin{tabular}{l|P{1.5cm}P{1.5cm}P{1.5cm}|P{1.5cm}P{1.5cm}}
\toprule
& \multicolumn{3}{c|}{\textbf{No. of Queries}} &  \multicolumn{2}{c}{\textbf{No. of Items}} \\
\textbf{Dataset} & \textbf{Observed} & \textbf{Novel} & \textbf{Gen} & \textbf{Observed} & \textbf{Novel} \\
\midrule
KeywordPrediction-10M & 220,845,427  & 5,022,052 & 5,022,052 & 4,999,996 & 5,435,443 \\
LF-AOL-270K-10 & 3,689,542 & 68,491 & 519,352 & 245,543 & 27,282 \\
LF-WikiHierarchy-550K-10 & 1,587,567 & 339,086 & 397,870 & 494,733 & 54,970 \\
LF-AmazonTitles-1.3M-10 & 2,225,354 & 624,830 & 970,237 & 1,174,739 & 130,526  \\
LF-Wikipedia-500K-10 & 1,781,890 & 271,620 & 783,743 & 450,963 & 50,107 \\
\bottomrule
\end{tabular}
\end{table}

%% file: tables/appendix_full_zero_shot_results.tex
\begin{table*}[!t]
\centering
\caption{Zero Shot Accuracies of different encoders when combined with \alg. Averaged across base encoders and datasets, \alg improves P@1, P@5, and R@10 by 9\%, 4.2\%, and 10.1\%, respectively.}
\label{tab:zsxc-complete-acc}
\fontsize{8}{8}\selectfont
\begin{tabular}{lcccccccc}
\toprule
\textbf{Model} &  \multicolumn{8}{c}{\textbf{LF-AOL-270K-10}} \\
\cmidrule{2-9}
& R@3 & R@5 & R@10 & R@30 & R@100 & P@1 & P@3 & P@5 \\
\midrule
NGAME & 43.90 & 48.80 & 54.20 & 60.14 & 65.39 & 30.90 & 15.43 & 10.32 \\
NGAME+\alg & 49.59 & 54.40 & 59.57 & 65.78 & 71.12 & 36.47 & 17.39 & 11.46 \\
\cmidrule{1-9}

ANCE & 51.81 & 59.63 & 67.84 & 77.38 & 84.78 & 33.43 & 18.10 & 12.52 \\
ANCE+\alg  & 53.44 & 60.29 & 67.82 & 76.47 & 83.23 & 36.84 & 18.67 & 12.66 \\
\cmidrule{1-9}
MACLR & 13.97 & 15.63 & 18.24 & 21.46 & 30.73 & 11.31 & 4.96 & 3.33  \\
MACLR+\alg & 48.15 & 54.20 & 61.29 & 70.54 & 78.86 & 34.32 & 16.89 & 11.41  \\
\cmidrule{1-9}
DPR  & 43.44 & 48.53 & 53.82 & 59.71 & 64.84 & 30.38 & 15.26 & 10.24 \\
DPR+\alg  & 50.14 & 55.13 & 60.22 & 66.18 & 71.49 & 36.80 & 17.57 & 11.61 \\

\midrule
\midrule
TF-IDF  & 20.10 & 23.76 & 28.05 & 34.10 & 39.13 & 13.74 & 7.27 & 5.08 \\
Zest-XML & 22.73 & 23.49 & 25.91 & 29.28 & 35.1 & 9.34 & 10.21 & 10.79 \\
Adam & 30.49 & 33.96 & 38.53 & 45.81 & 56.31 & 23.02 & 10.80 & 7.22 \\
SemSup-XC & 31.50 & 33.8 & 36.31 & 39.93 & 42.92 & 26.27 & 11.16 & 7.19 \\
DEXA & 31.64 & 36.50 & 41.85 & 48.81 & 55.59 & 21.68 & 11.16 & 7.73 \\
\midrule

\midrule
\textbf{Model} &  \multicolumn{8}{c}{\textbf{LF-WikiHierarchy-550K-10}} \\
\cmidrule{2-9}
& R@3 & R@5 & R@10 & R@30 & R@100 & P@1 & P@3 & P@5 \\
\midrule
NGAME & 37.96 & 47.24 & 58.66 & 73.67 & 84.31 & 46.01 & 32.93 & 25.63 \\
NGAME+\alg & 59.02 & 70.42 & 80.40 & 88.83 & 93.85 & 69.29 & 50.98 & 38.81 \\
\cmidrule{1-9}

ANCE & 35.97 & 44.92 & 56.28 & 72.30 & 85.56 & 43.06 & 30.68 & 23.98 \\
ANCE+\alg & 59.35 & 71.50 & 82.10 & 90.70 & 95.29 & 66.54 & 50.52 & 38.87 \\
\cmidrule{1-9}
MACLR & 21.97 & 27.33 & 35.47 & 44.33 & 66.89 & 30.37 & 19.07 & 14.44 \\
MACLR+\alg & 60.11 & 71.76 & 81.43 & 89.95 & 95.02 & 69.45 & 51.76 & 39.41 \\ 
\cmidrule{1-9}
DPR  & 37.14 & 47.34 & 59.29 & 74.6 & 85.64 & 44.84 & 32.3 & 25.53 \\
DPR+\alg  & 59.42 & 70.45 & 80.01 & 88.72 & 93.91 & 69.65 & 51.23 & 38.78 \\

\midrule
\midrule
TF-IDF  & 14.46 & 17.19 & 21.44 & 29.57 & 39.14 & 22.79 & 12.53 & 9.08 \\
Zest-XML & 14.56 & 16.39 & 17.48 & 29.28 & 22.85 & 13.97 & 13.29 & 12.67 \\
Adam & 29.31 & 36.08 & 45.04 & 58.39 & 71.76 & 38.3 & 25.19 & 19.13 \\
SemSup-XC & 40.56 & 44.39 & 46.81 & 48.95 & 50.09 & 57.45 & 36.11 & 24.67 \\
DEXA & 45.96 & 55.68 & 66.89 & 79.3 & 87.47 & 54.83 & 39.59 & 30.29 \\
\midrule

\midrule
\textbf{Model} &  \multicolumn{8}{c}{\textbf{LF-AmazonTitles-1.3M-10}} \\
\cmidrule{2-9}
& R@3 & R@5 & R@10 & R@30 & R@100 & P@1 & P@3 & P@5 \\
\midrule
NGAME & 24.20 & 29.40 & 36.44 & 46.71 & 56.04 & 30.42 & 19.94 & 15.38 \\
NGAME+\alg & 25.36 & 31.14 & 38.83 & 49.86 & 59.36 & 31.56 & 21.28 & 16.57 \\
\cmidrule{1-9}
ANCE & 18.14 & 23.32 & 30.72 & 43.66 & 57.76 & 22.38 & 15.14 & 12.02 \\
ANCE+\alg & 19.55 & 24.93 & 32.72 & 45.48 & 58.56 & 22.75 & 16.21 & 13.02 \\
\cmidrule{1-9}
MACLR & 17.59 & 21.96 & 28.59 & 40.35 & 53.38 & 21.93 & 14.50 & 11.39 \\
MACLR+\alg & 17.77 & 22.31 & 28.77 & 39.18 & 49.89 & 21.56 & 14.82 & 11.71 \\   
\cmidrule{1-9}
DPR & 25.72 & 32.07 & 40.98 & 53.87 & 64.18 & 31.10 & 21.29 & 16.82 \\
DPR+\alg & 25.40 & 31.60 & 40.31 & 52.63 & 62.18 & 30.49 & 21.04 & 16.62 \\

\midrule
\midrule
TF-IDF & 8.12 & 10.18 & 13.30 & 18.92 & 25.63 & 24.15 & 18.31 & 15.04 \\
Zest-XML & 5.42 & 6.35 & 6.87 & 7.89 & 8.67 & 5.58 & 4.71 & 4.22  \\
SemSup-XC & 11.28 & 14.73 & 20.04 & 29.26 & 38.27 & 11.68 & 8.41 & 6.85  \\
DEXA & 23.17 & 28.27 & 35.19 & 45.29 & 54.27 & 28.83 & 18.98 & 14.66 \\
\midrule

\midrule
\textbf{Model} &  \multicolumn{8}{c}{\textbf{LF-Wikipedia-500K-10}} \\
\cmidrule{2-9}
& R@3 & R@5 & R@10 & R@30 & R@100 & P@1 & P@3 & P@5 \\
\midrule
NGAME & 53.55 & 58.91 & 65.27 & 74.33 & 82.07 & 46.96 & 22.56 & 15.10\\
NGAME+\alg & 54.19 & 60.59 & 67.79 & 76.63 & 83.52 & 44.91 & 22.90 & 15.60 \\
\cmidrule{1-9}

ANCE & 40.46 & 47.71 & 58.91 & 75.70 & 87.98 & 30.67 & 16.57 & 11.92 \\  
ANCE+\alg & 54.84 & 62.59 & 71.59 & 82.91 & 91.16 & 41.59 & 23.02 & 16.05 \\
\cmidrule{1-9}
MACLR & 51.02 & 58.91 & 68.53 & 81.3 & 91.7 & 39.56 & 21.37 & 15.05 \\
MACLR+\alg  & 56.38 & 64.13 & 73.05 & 83.63 & 91.39 & 44.64 & 23.76 & 16.48  \\
\cmidrule{1-9}
DPR  & 55.28 & 62.70 & 71.20 & 81.22 & 89.43 & 42.90 & 23.02 & 15.96 \\
DPR+\alg  & 54.55 & 62.08 & 70.50 & 80.45 & 88.03 & 42.19 & 22.83 & 15.87 \\

\midrule
\midrule
TF-IDF  & 15.43 & 18.64 & 23.89 & 34.66 & 48.24 & 11.53 & 6.29 & 4.60 \\ 
Zest-XML & 6.86 & 9.95 & 14.73 & 21.63 & 25.55 & 2.62 & 2.62 & 2.32 \\
SemSup-XC & 50.38 & 53.80 & 57.08 & 60.01 & 61.03 & 46.60 & 21.46 & 13.90 \\
DEXA & 52.45 & 59.15 & 67.37 & 78.04 & 87.36 & 42.76 & 21.91 & 15.07 \\
\midrule

\end{tabular}
\end{table*}

%% file: tables/appendix_full_generalized_zero_shot_results.tex
\begin{table*}[!t]
\centering
\caption{Generalized Zero-Shot Accuracies of different encoders when combined with \alg. Averaged across base encoders and datasets, \alg improves P@1, P@5, and R@10 by 14.9\%, 10.4\%, and 9.8\%, respectively.}
\label{tab:gzsxc-complete-acc}
\fontsize{8}{8}\selectfont
\begin{tabular}{lcccccccc}
\toprule
\textbf{Model} &  \multicolumn{8}{c}{\textbf{LF-AOL-270K-10}} \\
\cmidrule{2-9}
& R@3 & R@5 & R@10 & R@30 & R@100 & P@1 & P@3 & P@5 \\
\midrule
NGAME & 24.71 & 30.70 & 38.27 & 47.93 & 55.24 & 20.16 & 13.81 & 10.43 \\
NGAME+\alg & 39.31 & 45.15 & 52.30 & 61.42 & 68.87 & 35.11 & 20.43 & 14.43 \\
\cmidrule{1-9}
ANCE & 31.06 & 39.29 & 49.72 & 58.74 & 74.00 & 22.63 & 15.86 & 12.25 \\
ANCE+\alg & 36.93 & 43.53 & 51.75 & 63.04 & 72.89 & 30.84 & 18.67 & 13.78 \\
\cmidrule{1-9}

MACLR  & 5.87 & 6.55 & 7.52 & 8.62 & 12.36 & 9.26 & 4.22 & 2.83 \\
MACLR+\alg  & 32.30 & 37.62 & 44.71 & 55.32 & 65.59 & 30.40 & 17.16 & 12.25 \\
\cmidrule{1-9}

DPR  & 24.35 & 30.37 & 37.99 & 44.50 & 54.95 & 19.71 & 13.61 & 10.31 \\
DPR+\alg & 39.51 & 45.43 & 52.57 & 61.73  & 69.13  & 35.07 & 20.54 & 14.52 \\

\midrule
TF-IDF  & 5.81 & 7.40 & 9.90 & 14.26 & 19.88 & 6.61 & 4.14 & 3.16 \\
Zest-XML & 19.67 & 22.18 & 26.57 & 29.8 & 33.39 & 26.34 & 14.71 & 10.19 \\
SemSup-XC & 20.98 & 22.55 & 23.92 & 25.21 & 25.86 & 26.12 & 13.97 & 9.08 \\
DEXA & 31.50 & 38.51 & 46.76 & 52.98 & 62.96 & 25.09 & 16.35 & 12.35 \\
\midrule
\midrule
\textbf{Model} &  \multicolumn{8}{c}{\textbf{LF-WikiHierarchy-550K-10}} \\
\cmidrule{2-9}
& R@3 & R@5 & R@10 & R@30 & R@100 & P@1 & P@3 & P@5 \\
\midrule
NGAME & 11.24 & 16.70 & 27.08 & 39.85 & 68.14 & 66.19 & 62.64 & 59.25\\
NGAME+\alg & 15.25 & 23.32 & 40.09 & 72.71 & 87.81 & 91.33 & 89.52 & 86.67 \\
\cmidrule{1-9}
ANCE & 11.02 & 16.36 & 25.89 & 37.67 & 65.11 & 68.76 & 63.48 & 58.87 \\
ANCE+\alg & 15.06 & 23.06 & 39.74 & 72.87 & 89.43 & 90.72 & 88.77 & 85.97 \\
\cmidrule{1-9}

MACLR  & 7.23 & 9.89 & 14.31 & 19.63 & 36.31 & 59.44 & 47.20 & 39.93 \\
MACLR+\alg & 14.48 & 22.33 & 38.37 & 70.92 & 87.89 & 88.81 & 87.19 & 84.60 \\
\cmidrule{1-9}
DPR  & 11.08 & 16.45 & 26.73 & 39.47 & 68.76 & 65.19 & 61.53 & 58.14 \\
DPR+\alg & 14.94 & 23.01 & 39.84 & 72.74 & 88.02 & 89.52 & 87.91 & 85.56 \\

\midrule
TF-IDF  & 6.67 & 8.36 & 10.88 & 15.56 & 22.16 & 64.50 & 42.24 & 32.1 \\
Zest-XML & 11.76 & 16.87 & 22.13 & 26.68 & 29.83 & 68.86 & 49.19 & 38.73 \\
SemSup-XC & 13.80 & 19.93 & 28.37 & 32.15 & 32.38 & 90.51 & 84.87 & 78.61 \\
DEXA & 13.74 & 21.12 & 36.17 & 53.96 & 81.81 & 76.18 & 76.94 & 75.38 \\
\midrule
\midrule
\textbf{Model} &  \multicolumn{8}{c}{\textbf{LF-AmazonTitles-1.3M-10}} \\
\cmidrule{2-9}
& R@3 & R@5 & R@10 & R@30 & R@100 & P@1 & P@3 & P@5 \\
\midrule
NGAME & 17.45 & 22.58 & 30.25 & 42.72 & 54.81 & 45.14 & 39.15 & 34.72 \\
NGAME+\alg & 17.69 & 23.17 & 31.49 & 45.09 & 58.19 & 47.77 & 42.68 & 38.35 \\
\cmidrule{1-9}
ANCE & 9.45 & 12.41 & 17.31 & 27.25 & 40.62 & 27.65 & 22.76 & 19.76 \\
ANCE+\alg & 11.00 & 15.25 & 22.34 & 35.85 & 51.26 & 36.78 & 32.41 & 29.31 \\
\cmidrule{1-9}

MACLR & 9.13 & 11.76 & 15.99 & 24.57 & 36.57 & 27.50 & 21.86 & 18.60 \\
MACLR+\alg  & 9.57 & 13.05 & 18.85 & 30.13 & 43.61 & 31.49 & 26.92 & 23.97 \\
\cmidrule{1-9}
DPR  & 14.85 & 19.48 & 26.93 & 35.42 & 55.12 & 38.18 & 32.89 & 29.20 \\
DPR+\alg  & 15.38 & 20.73 & 29.25 & 44.28 & 59.42 & 43.08 & 38.66 & 35.00 \\

\midrule
TF-IDF  & 13.71 & 16.47 & 20.40 & 27.08 & 34.81 & 16.33 & 9.83 & 7.35 \\
Zest-XML & 12.34 & 14.45 & 22.87 & 26.79 & 39.71 & 41.36 & 33.7 & 28.29 \\
SemSup-XC & 7.76 & 10.59 & 15.21 & 23.41 & 30.87 & 25.13 & 20.93 & 18.37 \\
DEXA & 18.08 & 23.27 & 30.89 & 38.70 & 54.21 & 48.19 & 40.45 & 35.47 \\
\midrule
\midrule
\textbf{Model} &  \multicolumn{8}{c}{\textbf{LF-Wikipedia-500K-10}} \\
\cmidrule{2-9}
& R@3 & R@5 & R@10 & R@30 & R@100 & P@1 & P@3 & P@5 \\
\midrule
NGAME & 52.24 & 60.96 & 69.58 & 78.67 & 85.50 & 81.86 & 60.13 & 45.38 \\
NGAME+\alg & 50.29 & 59.56 & 69.27 & 79.78 & 87.27 & 78.99 & 58.60 & 44.86 \\
\cmidrule{1-9}
ANCE & 29.66 & 35.51 & 43.39 & 56.22 & 71.99 & 42.91 & 27.54 & 20.92 \\
ANCE+\alg & 43.91 & 52.88 & 63.46 & 76.42 & 86.32 & 71.39 & 49.56 & 38.09 \\
\cmidrule{1-9}

MACLR & 29.20 & 36.38 & 46.62 & 61.86 & 75.69 & 46.59 & 31.12 & 24.36 \\
MACLR+\alg & 42.80 & 51.95 & 62.82 & 75.93 & 85.77 & 70.52 & 51.10 & 39.24 \\
\cmidrule{1-9}
DPR  & 38.93 & 49.21 & 61.84 & 72.17 & 87.01 & 51.54 & 40.30 & 32.71 \\
DPR+\alg  & 44.88 & 54.95 & 66.71 & 80.17 & 89.10 & 70.39 & 52.50 & 40.91 \\

\midrule
TF-IDF  & 9.49 & 11.67 & 14.79 & 20.87 & 30.10 & 15.07 & 9.19 & 6.93 \\
Zest-XML & 32.24 & 38.01 & 45.15 & 54.31 & 59.32 & 60.16 & 39.33 & 29.21 \\
SemSup-XC & 23.44 & 28.83 & 38.08 & 48.01 & 61.30 & 54.20 & 40.72 & 29.58 \\
DEXA & 45.83 & 55.05 & 65.86 & 74.63 & 88.11 & 67.98 & 48.88 & 37.90 \\
\midrule
\end{tabular}
\end{table*}

%% file: tables/appendix_full_ablations_numbers.tex
\begin{table*}[!t]
\centering
\fontsize{7}{12}\selectfont
\caption{Ablation study on meta-classifier generator $\cG$ and classifier selector $\cS$ component in \alg on LF-WikiHierarchy-550K-10 dataset. Here \alg is NGAME+\alg}
\label{tab:ablations-full}
    \begin{tabular}{l|cccccccc|cccccccc}
    \hline
    Method & \multicolumn{8}{c|}{Zero shot} & \multicolumn{8}{c}{Generalized} \\
    \hline
    & \textbf{R@3} & \textbf{R@5} & \textbf{R@10} & \textbf{R@30} & \textbf{R@100} & \textbf{P@1} & \textbf{P@3} & \textbf{P@5} & \textbf{R@3} & \textbf{R@5} & \textbf{R@10} & \textbf{R@30} & \textbf{R@100} & \textbf{P@1} & \textbf{P@3} & \textbf{P@5} \\
    \hline \hline
    \alg ($\cD=1$, $K=3$) & 59.02 & 70.42 & 80.40 & 88.83 & 93.85 & 69.29 & 50.98 & 38.81 & 15.25 & 23.32 & 40.09 & 72.71 & 87.81 & 91.33 & 89.52 & 86.67 \\
\hline
    $\cD=2$ & 59.83 & 70.83 & 80.39 & 88.60 & 93.79 & 70.36 & 51.66 & 39.06 & 15.35 & 23.49 & 40.48 & 73.51 & 88.35 & 92.13 & 90.10 & 87.24 \\
    $\cD=4$ & 60.00 & 70.85 & 80.27 & 88.51 & 93.70 & 70.71 & 51.82 & 39.11 & 15.37 & 23.52 & 40.56 & 73.72 & 88.42 & 92.28 & 90.18 & 87.33 \\
\hline
    $\cG$ as Sum & 37.55 & 47.10 & 59.01 & 73.87 & 84.80 & 45.49 & 32.59 & 25.46 & 11.23 & 16.84 & 27.15 & 39.57 & 56.20 & 68.11 & 64.01 & 60.51 \\ 
    $\cG$ as wt. Sum & 38.34 & 47.72 & 59.29 & 74.33 & 85.12 & 46.39 & 33.23 & 25.88 & 11.29 & 16.80 & 27.28 & 40.25 & 57.74 & 66.60 & 63.08 & 59.72 \\
\hline
    \alg + Enc. Embed. & 49.10 & 60.60 & 72.36 & 82.31 & 87.66 & 56.79 & 42.14 & 32.92 & 14.25 & 21.49 & 35.89 & 53.75 & 73.33 & 87.17 & 83.12 & 79.29 \\
\hline
    $K=1$ & 58.82 & 70.24 & 80.11 & 88.79 & 93.84 & 68.98 & 50.79 & 38.56 & 15.13 & 23.13 & 39.67 & 72.01 & 87.59 & 91.04 & 88.85 & 85.95 \\
$K=2$ & 59.18 & 70.33 & 80.25 & 88.94 & 93.94 & 69.34 & 51.01 & 38.74 & 15.23 & 23.27 & 39.97 & 72.41 & 87.74 & 91.15 & 89.30 & 86.40 \\
    $K=6$ & 59.63 & 70.83 & 80.79 & 88.94 & 93.81 & 69.94 & 51.52 & 39.12 & 15.09 & 23.14 & 39.93 & 72.98 & 88.04 & 89.93 & 88.59 & 86.01 \\
    $K=20$ & 58.67 & 69.81 & 79.80 & 88.06 & 93.30 & 69.07 & 50.78 & 38.57 & 15.23 & 23.27 & 39.93 & 72.75 & 87.92 & 91.76 & 89.61 & 86.64 \\
    \hline
\end{tabular}
\end{table*}

%% file: sections/9_ethics.tex
\section*{Ethical Considerations}
\label{app:ethics}
Our data usage and service provision align with the approvals of our legal and ethical committees. Sponsored search plays a vital role in guaranteeing free and convenient access to information for billions of users worldwide, thereby sustaining free web engines. Our research contributes to the enhancement of the online search experience, offering accessible information to millions of users around the world while enabling advertisers to connect with relevant audiences. This mutual relationship cultivates a more effective and valuable advertising ecosystem, benefiting both users and advertisers. Furthermore, our endeavors support the financial growth of small and medium businesses, including local establishments, by increasing revenue, expanding market reach, and reducing costs associated with reaching new customers.